\documentclass[11pt]{llncs}

\usepackage{amssymb}
\usepackage{amsmath}
\usepackage{amsthm}
\usepackage{textcomp}
\usepackage{array}
\usepackage{multirow}
\usepackage{comment}

\usepackage{color}
\usepackage[left=3cm,top=3cm,right=3cm,bottom=3cm]{geometry}

\newcommand{\bra}[1]{\langle #1|}
\newcommand{\ket}[1]{|#1\rangle}

\newcommand{\cent}[0]{\mbox{\textcent}}
\newcommand{\dollar}[0]{\$}
\newcommand{\lrhd}[0]{ \lhd \mspace{-3mu} \rhd }

\newtheorem{fact}{Fact}

\title{Proving the power of postselection\thanks{Earlier versions of this paper appeared as \cite{YS10D,YS11B}}}
\author{Abuzer Yakary{\i}lmaz$ ^{\mbox{\tiny 1,}} $\thanks{Yakary{\i}lmaz was partially supported by the Scientific and Technological Research Council of Turkey (T\"{U}B\.ITAK) with grant 108E142 and FP7 FET-Open project QCS.}
\and
A.C. Cem Say$ ^{\mbox{\tiny 2,}} $\thanks{Say's work was partially supported by the Scientific and Technological
Research Council of Turkey (T\"{U}B\.ITAK) with grant 108E142.}}

\institute{$ ^{\mbox{\tiny 1}} $University of Latvia, Faculty of Computing, Raina bulv. 19, Riga, LV-1586, Latvia \\
\email{abuzer@lu.lv} ~~ \\
$ ^{\mbox{\tiny 2}} $Bo\u{g}azi\c{c}i University, Department of Computer Engineering, Bebek 34342 \.{I}stanbul, Turkey
\\ \email{say@boun.edu.tr} ~~
 \\~~\\
\today
}

\begin{document}

\newlength{\twidth}
\maketitle
\pagenumbering{arabic}

%-----------------------------------------------------------------------------%
\begin{abstract} \label{abstract:Abstract}
%-----------------------------------------------------------------------------%
It is a widely believed, though unproven, conjecture that the capability of postselection increases the language recognition power of both probabilistic and quantum polynomial-time computers. It is also unknown whether polynomial-time quantum machines with postselection are more powerful than their probabilistic counterparts with the same resource restrictions. We approach these problems by imposing additional constraints on the resources to be used by the computer, and are able to prove for the first time that postselection does augment the computational power of both classical and quantum computers, and that quantum does outperform probabilistic in this context, under simultaneous time and space bounds in a certain range. We also look at postselected versions of space-bounded classes, as well as those corresponding to  error-free and one-sided error recognition, and provide classical characterizations. It is shown that $\mathsf{NL}$ would equal $\mathsf{RL}$ if the randomized machines had the postselection capability.
\\ \\
\textbf{keywords:}
postselection, quantum Turing machines, probabilistic Turing machines, space-bounded computation,
	one-sided error, zero error
\end{abstract}

% SSSSSSSSSSSSSSSSSSSSSSSSSSSSSSSSSSSSSSSSSSSSSSSSSSSSSSSSSSSSSSSSSSSSSSSSSSSSSSSS %
% SSSSSSSSSSSSSSSSSSSSSSSSSSSSSSSSSSSSSSSSSSSSSSSSSSSSSSSSSSSSSSSSSSSSSSSSSSSSSSSS %
% SSSSSSSSSSSSSSSSSSSSSSSSSSSSSSSSSSSSSSSSSSSSSSSSSSSSSSSSSSSSSSSSSSSSSSSSSSSSSSSS %
\section{Introduction} \label{sec:Introduction}
% SSSSSSSSSSSSSSSSSSSSSSSSSSSSSSSSSSSSSSSSSSSSSSSSSSSSSSSSSSSSSSSSSSSSSSSSSSSSSSSS %
% SSSSSSSSSSSSSSSSSSSSSSSSSSSSSSSSSSSSSSSSSSSSSSSSSSSSSSSSSSSSSSSSSSSSSSSSSSSSSSSS %
% SSSSSSSSSSSSSSSSSSSSSSSSSSSSSSSSSSSSSSSSSSSSSSSSSSSSSSSSSSSSSSSSSSSSSSSSSSSSSSSS %

The notion of postselection as a mode of computation was introduced by
Aaronson \cite{Aa05}. Postselection is the (unrealistic) capability
of discarding all branches of a computation in which a specific event
does not occur, and focusing on the surviving branches for the final
decision about the membership of the input string in the recognized
language. Aaronson examined $\mathsf{PostBQP}$, the class of languages recognized with bounded error
by polynomial-time quantum computers with postselection, and showed it
to be identical to the well-known classical complexity class $\mathsf{PP}$. The corresponding class for probabilistic polynomial-time computers with postselection is known to equal $\mathsf{BPP_{path}}$. It is, however, still an open question whether postselection adds anything to the power of quantum or classical polynomial-time computation, since we do not know whether the standard classes (without postselection) for these models, that is, $\mathsf{BQP}$ and $\mathsf{BPP}$, respectively, equal their postselected versions or not. It is also not known whether $\mathsf{BPP_{path}}=\mathsf{PP}$, or polynomial-time quantum machines with postselection outperform their classical counterparts, as is conjectured to be the case between the standard versions. All these classes sit between $\mathsf{P}$ and $\mathsf{PSPACE}$, and any proof of unequality between them would therefore be as hard as proving $\mathsf{P}\neq \mathsf{PSPACE}$. 

In this paper, we approach the problem of evaluating the power of postselection by imposing more resource restrictions on the computer. We demonstrate certain simultaneous space and time bounds, under which it is proven that postselection increases the power of both probabilistic and quantum machines, and that quantum computers with this capability outperform classical ones. More precisely, we prove that polynomial-time and $o(\log \log n)$-space probabilistic computers with postselection are more powerful than their standard versions. To handle  quantum machines, we constrain the time bound further to allow only real-time computation, where the computer has just time enough for a single left-to-right scan of the input. Under this restriction, we prove that quantum computers with postselection that are allowed to use only $O(1)$ space outperform those without postselection, and that probabilistic machines with postselection 
are inferior to quantum ones for any common space bound that is sublogarithmic.

We also examine  postselected versions of classes that are only space-bounded, as well as those corresponding to  error-free and one-sided error recognition, and provide several classical characterizations. It turns out that postselection adds nothing to purely space-bounded machines, and the 
error-free and one-sided error classes for probabilistic machines with postselection for any space bound $s(n)=\Omega(\log n)$ are identical to each other, and to $\mathsf{NSPACE}(s(n))$. We show that allowing one-sided bounded error to be committed by a small-space quantum machine with postselection enlarges the class of recognized languages, in contrast to the probabilistic case. We also show that $\mathsf{NL}$ would equal $\mathsf{RL}$ if the randomized machines had the postselection capability.

Aaronson's definition is based on quantum circuits, whereas we use a model of postselection based on Turing machines (TMs), which facilitates dealing with space bounds. L\={a}ce \textit{et al}.  were the first to use machines, rather than circuits, for modeling postselection, in their groundbreaking work \cite{LSF09} on the effect of this capability on quantum finite
automata. We show that an idiosyncratic feature in their model makes it strictly more powerful than ours, which is more in alignment with  Aaronson's original
definition.

A key tool used in our proofs is the \textit{Turing machine with restart}, which is simply an ordinary TM that has been augmented with the capability of resetting itself to the initial configuration in a single step. We demonstrate that this restart action is equivalent in effect to postselection. The model of TM with restart is obtained by generalizing the
recently introduced real-time finite automata with restart \cite{YS10B}.

The rest of the paper is structured as follows: Section \ref{sec:Prel} recalls the standard definitions of probabilistic and quantum machines, and introduces the notation that will be used later. Machines with restart are defined in Section \ref{sec:restart}. Section \ref{sec:Posdefs} contains our definition of postselection. Our proofs of the superiority of machines with postselection over their standard versions, and that of the quantum variant over the classical one, are presented in Section \ref{sec:powerofpost}. Characterizations of classes of languages recognized by these machines exactly, or with one-sided error, can be found in Section \ref{sec:errorfree}. The model of L\={a}ce \textit{et al}. is compared with ours in Section \ref{sec:LPostFA}. Section \ref{sec:ConcludingRemarks} is a conclusion. Additional details on some of the discussed points are provided in the Appendices.

% SSSSSSSSSSSSSSSSSSSSSSSSSSSSSSSSSSSSSSSSSSSSSSSSSSSSSSSSSSSSSSSSSSSSSSSSSSSSSSSS %
% SSSSSSSSSSSSSSSSSSSSSSSSSSSSSSSSSSSSSSSSSSSSSSSSSSSSSSSSSSSSSSSSSSSSSSSSSSSSSSSS %
% SSSSSSSSSSSSSSSSSSSSSSSSSSSSSSSSSSSSSSSSSSSSSSSSSSSSSSSSSSSSSSSSSSSSSSSSSSSSSSSS %
\section{Preliminaries} \label{sec:Prel}
% SSSSSSSSSSSSSSSSSSSSSSSSSSSSSSSSSSSSSSSSSSSSSSSSSSSSSSSSSSSSSSSSSSSSSSSSSSSSSSSS %
% SSSSSSSSSSSSSSSSSSSSSSSSSSSSSSSSSSSSSSSSSSSSSSSSSSSSSSSSSSSSSSSSSSSSSSSSSSSSSSSS %
% SSSSSSSSSSSSSSSSSSSSSSSSSSSSSSSSSSSSSSSSSSSSSSSSSSSSSSSSSSSSSSSSSSSSSSSSSSSSSSSS %

Since our discussion will involve both space and time complexity issues, we will use resource-bounded
versions of Turing machines  as our models of computation. We now present quick definitions of standard probabilistic and quantum TMs, which will serve as templates for the new models to be introduced in the subsequent sections. We assume some familiarity with probabilistic and quantum computation, and the reader is referred to \cite{YS11A}  for a wider coverage of the basics.

All TMs we will consider have a read-only input tape, which contains the input string, sandwiched between the two special end-marker symbols $ \cent $ and $ \dollar $. The two-way input tape head is initially positioned on the left end-marker $ \cent $, and is not allowed to leave the area delimited by the end-markers. Space complexity is measured by the maximum number of cells that are ever visited with non-zero probability (or amplitude, in the quantum case) by the read/write head of the single work tape of the machine, as a function of the length of the input string. A configuration of a TM is a collection of its internal state, positions of the input and work tape heads, and  the contents of the work tape.

A probabilistic Turing machine (PTM) is a 6-tuple
\begin{equation*}
      \mathcal{P}=(Q,\Sigma,\Gamma,\delta,q_{1},\Delta),
\end{equation*}
where $Q$, $\Sigma$, $\Gamma$, and $ q_{1} $ denote the set of internal states, the input alphabet,  the work tape alphabet, and the initial state, respectively.

The transition function $ \delta $ is specified such that
\begin{equation*}
       \delta(q,\sigma,\gamma,q^{\prime},d_{i},\gamma^{\prime},d_{w}) \in \tilde{\mathbb{R}}
\end{equation*}
is the probability that the PTM will change its internal state to $
q^{\prime} $, write $ \gamma^{\prime} $ on the work tape,
and update the positions of the
input and work tape heads with respect to $ d_{i} $ and $ d_{w} $,
respectively,
where $ d_{i},d_{w} \in \lrhd = \{left,right,stationary\} $,  if it
scans $ \sigma $ and $ \gamma $ on the input and work tapes,
respectively, when originally in internal state $ q $. $ \tilde{\mathbb{R} } $
is the set consisting of $ p \in \mathbb{R} $ such that there
is a deterministic
algorithm that computes $ p $ to within $2^{-n}$ in time polynomial in $n$.

For each input string $ w \in \Sigma^{*} $, $ \delta $ 
defines a unique configuration transition matrix, $ A^{w}
$.
A PTM is \textit{well-formed}  if all columns of $ A^{w} $ are stochastic vectors.
This constraint defines the following local conditions for PTM
well-formedness
that $ \delta $ must obey:
For each $ q \in Q $, $ \sigma \in \tilde{\Sigma}  = \Sigma \cup \{ \cent, \dollar \} $,  and $ \gamma \in \Gamma $,
\begin{equation*}
      \sum_{q^{\prime},d_{i},\gamma^{\prime},d_{w}}
\delta(q,\sigma,\gamma,q^{\prime},d_{i},\gamma^{\prime},d_{w}) = 1,
\end{equation*}
where $ q^{\prime} \in Q $, $ \gamma^{\prime} \in \Gamma $, and $ d_{i},d_{w} \in \lrhd$.

In order to make the presentation of the essential differences among the various machine models to be discussed in the paper easier, we decree that all TM definitions include the item  $ \Delta  = \{ \tau_1, \ldots, \tau_k \} $, which is the set of ``move outcomes," that summarize the overall condition of the computation after each  step. In standard (probabilistic and quantum) TMs, $ \Delta = \{ c,a,r \} $, and $Q$, the usual finite set of internal states of the machine, is partitioned to three corresponding subsets 
$ Q_c $, $ Q_a $, and $ Q_r $, called the sets of continuing (non-halting), accepting, and rejecting  states, respectively.
The computation is terminated, and the input is 
accepted (resp., rejected) if the TM enters a state belonging to $ Q_{a} $ (resp., $ Q_{r} $). The machine continues with the next move otherwise.
We use the name $ Q_{h} $ (the set of halting states) to refer to all states at which the computation is terminated, and we naturally have 
$ Q_{h} = Q_{a} \cup Q_{r} $ for all standard TM variants.

For ease in the modeling of space-efficient quantum computation, our quantum Turing machines (QTMs) are assumed to contain an additional component, namely, a finite register, which is used in the observation of the move outcomes. The set of different values that this register can contain is denoted by $ \Omega $, and is also partitioned into $ |\Delta| $ subsets 
($\Omega_{c} $, $ \Omega_{a}$, and $\Omega_{r} $ in the case of standard QTMs), corresponding to the different types of move outcomes.
A QTM is then defined as a 7-tuple
\begin{equation*}
      \mathcal{M}=(Q,\Sigma,\Gamma,\Omega,\delta,q_{1},\Delta),
\end{equation*}
where $\Omega_{c}$ is required to contain a special \textit{initial symbol} $ \omega_{1} $, and  new conditions to be described below are imposed on the transition function $\delta$.

The transition function of a QTM is specified so that
\begin{equation*}
      \delta(q,\sigma,\gamma,q^{\prime},d_{i},\gamma^{\prime},d_{w},\omega)
\in \tilde{\mathbb{R}}
\end{equation*}
is the amplitude with which the QTM will change its internal state to
$ q^{\prime} $, write $ \gamma^{\prime} $ on the work tape
and $ \omega $ in the finite register, and update the positions of the
input and work tape heads with respect to $ d_{i} $ and $ d_{w} $,
respectively,
where $ d_{i},d_{w} \in \lrhd $, if it scans $ \sigma $ and $ \gamma $ on the
input and work tapes, respectively, when originally in internal state
$ q $. (The finite register always contains $ \omega_{1}$ at the beginning of every move.) 
 
After each transition, the finite register is measured to see which one of the sets $ \Omega_{a}$, $\Omega_{r} $, or $\Omega_{c} $ the current register symbol belongs to, and the
following actions are associated with the measurement outcomes:
\begin{itemize}
      \item ``$c$": the computation continues;
      \item ``$a$": the computation halts, and the input is accepted;
      \item ``$r$": the computation halts, and the input is rejected.
\end{itemize}
The finite register is irreversibly reinitialized to $ \omega_{1}$ before the next transition takes place.

Any superiority that quantum computers have over their probabilistic counterparts can be traced to the fact that machine configurations can have negative as well as positive amplitudes, sometimes allowing parallel computational branches to interfere with each other in a way that is impossible in classical computation. The amount with which a particular symbol will contribute to the probability of measurement of the associated outcome is the modulus squared of the corresponding amplitude at the time of observation. Appendix \ref{app:wellform} contains a description of the well-formedness conditions that quantum machines must satisfy.

Any sufficiently general quantum model can simulate the corresponding probabilistic model (subject to the same space and time restrictions) exactly, with essentially no overhead \cite{Wa09,YS11A}. So the question faced when comparing such a pair of models is always whether they are equivalent in power, or the quantum version can outperform the probabilistic one.

We will examine the language recognition of different types of machines under several error regimes. The terminology to be used in this regard is summarized below. 

The language $ L \subseteq \Sigma^{*} $ recognized by machine $
\mathcal{M} $ with (\textit{strict}) \textit{cutpoint}
$ \lambda \in \mathbb{R} $ is defined as
\begin{equation*}
       L = \{ w \in \Sigma^{*} \mid \mbox{P(}{\mathcal{M}} \mbox{ accepts } w) > \lambda \}.
\end{equation*}

The language $ L \subseteq \Sigma^{*} $ recognized by machine $
\mathcal{M} $ with \textit{nonstrict cutpoint}
$ \lambda \in \mathbb{R} $ is defined as \cite{BJKP05}
\begin{equation*}
       L = \{ w \in \Sigma^{*} \mid \mbox{P(}{\mathcal{M}} \mbox{ accepts } w)\geq \lambda \}.
\end{equation*}

The two cases described above comprise recognition with (\textit{two-sided}) \textit{unbounded error},
where every member of the recognized language $L$ is accepted
with a probability greater than every nonmember of $L$.

Machines that recognize a language with cutpoint 0, i.e. those that accept a string with nonzero probability if and only if that string is a member of the language, are said to be \textit{nondeterministic}.

The language $ L \subseteq \Sigma^{*} $ is said to be recognized by machine $
\mathcal{M} $ with (\textit{two-sided}) \textit{bounded error} if there exists an error bound $ \epsilon $
($ 0 \le \epsilon < \frac{1}{2} $) such that 
\begin{itemize}
       \item  P($\mathcal{M}$ accepts $w$) $\ge 1 - \epsilon $ for all $ w \in L $, and,
       \item P($\mathcal{M}$ rejects $w$) $\ge 1 - \epsilon $ for all $ w \notin L $.
\end{itemize}
Recognition with\textit{ one-sided bounded error} is defined as recognition by a nondeterministic machine with bounded error. A bounded-error machine whose error bound equals 0 is said to be performing \textit{error-free} (or \textit{exact}) computation.

Table \ref{tbl:std-class} lists some of the  language classes that we will mention. Most of the terminology here is standard, but our definition of the $\mathsf{EQSPACE}$ classes is different from that of Watrous, who used this designation for the first time \cite{Wa99} in terms of QTMs that were allowed to fail to halt with probability 1 for input strings that are not members of the language to be recognized. Our definition, paralleling the definition of $\mathsf{EQTime}$ in \cite{BV97}, corresponds to the class $\mathsf{EQ_{AS}SPACE}$ in \cite{Wa99}.  The reader should also note that the $\mathsf{EPTIME}(t)$ classes in Table \ref{tbl:std-class} clearly equal $\mathsf{TIME}(t)$ for any $t$; we include the name here only because the postselected version of these classes will be studied in Section \ref{sec:errorfree}.

\begin{table}[h!]
\caption{The standard classes of languages recognized by PTMs and QTMs}
\centering
\fbox{
\small
\begin{minipage}{0.97\textwidth}
\[
	\begin{array}{lccccc}
		\multicolumn{6}{c}{\textit{space-bounded classes}}
		\\
		\mbox{\underline{machine type}}  & \mbox{\underline{unbounded-error}} 
		& \mbox{\underline{nondeterministic}} &  \mbox{\underline{bounded-error}}
		&  \mbox{\underline{one-sided bounded-error}}
		&  \mbox{\underline{error-free}}
		\\
		\mbox{PTM} & \mathsf{PrSPACE} & \mathsf{NSPACE} & \mathsf{BPSPACE} & \mathsf{RSPACE} & \mathsf{EPSPACE}
		\\
		\mbox{QTM} & \mathsf{PrQSPACE} & \mathsf{NQSPACE} & \mathsf{BQSPACE} & \mathsf{RQSPACE} & \mathsf{EQSPACE}
		\\
		\multicolumn{6}{c}{\textit{time-bounded classes}}
		\\
		\mbox{\underline{machine type}}  & \mbox{\underline{unbounded-error}} 
		& \mbox{\underline{nondeterministic}} &  \mbox{\underline{bounded-error}}
		&  \mbox{\underline{one-sided bounded-error}}
		&  \mbox{\underline{error-free}}
		\\
		\mbox{PTM} & \mathsf{PrTIME} & \mathsf{NTIME} & \mathsf{BPTIME} & \mathsf{RTIME} & \mathsf{EPTIME}
		\\
		\mbox{QTM} & \mathsf{PrQTIME} & \mathsf{NQTIME} & \mathsf{BQTIME}  & \mathsf{RQTIME} & \mathsf{EQTIME}
	\end{array}
\]
\end{minipage}
}	
\label{tbl:std-class}
\end{table}

Simultaneous time-space bounds   for nondeterministic, probabilistic, and quantum machines have been studied in, for instance, \cite{BM80,DM06,MW08}, respectively, from which we take the following class definitions:

$\mathsf{BPTISP}(t,s)$, $\mathsf{RTISP}(t,s)$, and $\mathsf{NTISP}(t,s)$ are the classes of languages recognized with (two-sided) bounded error, one-sided bounded error, and with cutpoint zero, respectively,  by PTMs running in time $t$, and using space $s$. $\mathsf{BQTISP}$ is the quantum counterpart of $\mathsf{BPTISP}$. We define $\mathsf{PrTISP}$ and $\mathsf{PrQTISP}$ to be the unbounded-error counterparts of $\mathsf{BPTISP}$ and $\mathsf{BQTISP}$. 

We should clarify a potential source of confusion about randomized space-bounded classes \cite{Sa96}. $\mathsf{RSPACE}(s)$, which is defined to be the class of languages recognized with  one-sided error by  PTMs using space $s$, turns out to be identical to $\mathsf{NSPACE}(s)$, i.e. the class of languages recognized by nondeterministic TMs that use space $s$ \cite{Gi77}. The designation $\mathsf{RL}$, which was originally a shorthand for $\mathsf{RSPACE}(\log n)$, is now used to denote the more interesting class of languages that are recognized with positive one-sided error by logspace PTMs with polynomial time bounds \cite{Go08}, that is, $\mathsf{RTISP}(poly(n),\log n)$.

Some of our results involve \textit{real-time} machines, i.e., those that are restricted to move the input head to the right in every step of the computation, forcing them to have a runtime of $n+2$ steps, where $n$ is the length of the input.\footnote{It is well-known \cite{Ra63B,Aa74} that increasing the number of work tapes increases the language recognition power of standard versions of real-time machines. Our models have a single work tape, but all our results regarding real-time machines with postselection remain valid when multiple work tapes are allowed.} Inspired by Bruda's definition \cite{Br02} of the  $\mathsf{rt\mbox{-}SPACE}(s)$ classes, we define 
\textsc{rt}$\mathsf{BPSPACE}(s)$ as the class of languages recognized with (two-sided) bounded error 
by real-time PTMs using space $s$. 
\textsc{rt}$\mathsf{BQSPACE}$ is the quantum counterpart of \textsc{rt}$\mathsf{BPSPACE}$. 

It is well known that, for machines with constant space usage, one can remove the work tape altogether, at the cost of having a longer program, without changing the recognized language. This specialization of TMs yields the well-known finite automata \cite{Si06,Ra63,YS11A}. One quirk in the literature that we should be careful about is the fact that the transition probabilities and amplitudes of probabilistic and quantum finite automata (PFAs and QFAs, respectively) are allowed to be arbitrary real numbers (including uncomputable ones) of absolute value at most 1, enabling these machines to recognize many Turing-undecidable languages \cite{Ra63}, and we will use different  names (Table \ref{tbl:rt-class}) to denote the classes arising from this different range of $\delta$.

\begin{table}[h!]
\caption{Classes of languages recognized by real-time finite automata}
\centering
\fbox{
\begin{minipage}{0.7\textwidth}
\[
	\begin{array}{lcc}
		\mbox{\underline{machine type}}  & \mbox{\underline{unbounded-error}} 
		& \mbox{\underline{nondeterministic}} 
		\\
		\mbox{real-time PFA} & \mathsf{S}
			\cup \mathsf{coS} = \mathsf{uS}
			& \mathsf{REG}
		\\
		\mbox{real-time QFA} & \mathsf{QAL}
		\cup \mathsf{coQAL} = \mathsf{uQAL} & 
		\mathsf{NQAL}
		\\ \hline
	\end{array}	
\]
\footnotesize
$ \mathsf{S} $ and $ \mathsf{QAL} $ are the classes of languages recognized with cutpoint $ \frac{1}{2} $ by real-time PFAs \cite{Ra63} and QFAs \cite{YS11A}, respectively.
$ \mathsf{REG} $ is the class of regular languages.
\end{minipage}
}	
\label{tbl:rt-class}
\end{table}

% sssssssssssssssssssssssssssssssssssssssssssssssssssssssssssssssssssssssssssssssss %
% sssssssssssssssssssssssssssssssssssssssssssssssssssssssssssssssssssssssssssssssss %
\section{Probabilistic and quantum machines with restart} \label{sec:restart}
% sssssssssssssssssssssssssssssssssssssssssssssssssssssssssssssssssssssssssssssssss %
% sssssssssssssssssssssssssssssssssssssssssssssssssssssssssssssssssssssssssssssssss %

In this section, we introduce the effects of adding the capability of ``restarting" the computation, that is, restoring the internal state, input head position, and the work tape to their initial settings in a single move, to the set of allowed actions of several Turing machine variants and specializations. This simple and seemingly useless action turns out to be important in our analysis of postselection in the rest of the paper. Finite automata with restart were introduced and analyzed in \cite{YS10B}, we generalize the concept to TMs here.

Using our general framework of TM definitions mentioned in the previous section, we can define a (probabilistic or quantum) TM with restart simply by stating that the set of possible move outcomes, 
$ \Delta $,  contains an additional element, $ rs $,
and the overall set of states, $Q$, is now correspondingly partitioned into four subsets: The usual $ Q_c $, $ Q_a $, and $ Q_r $,  and the set $ Q_{rs} $, namely, the set of restarting states.
Any transition to a state in $ Q_{rs} $ results 
the machine to restart from the initial configuration in the next move, as explained above.

A segment of computation of a TM with restart which begins with a (re)start, and ends with
a halting or restarting state will be called a \textit{round}. For any time bound $ t $, a \textit{$ t $-time TM with restart} is  a TM with restart with the restriction that the runtime of no single round is greater than $ t $. We will be focusing on \textit{real-time machines with restart}, that is, TMs in which the input head is forbidden to make single-step leftward or stationary moves, but restarts are allowed. Note that the \textit{overall} expected runtime of a $ t $-time TM with restart can be much more than $ t $; for instance, \cite{YS10B} contains several examples of real-time TMs with restart that have exponential runtime.

Let $ p^{a}_{\mathcal{R}} (w) $ ($ p^{r}_{\mathcal{R}} (w) $) be the probability that $w$ is accepted (rejected) in a single round of a TM with restart named $ \mathcal{R} $. For a given input string $ w \in \Sigma^{*} $,
the overall acceptance and rejection probabilities of $w$ can be calculated as shown in the following lemma \cite{YS10B}.

\begin{lemma}
	\label{lem:overall-acc-rej}
	P($\mathcal{R}$ accepts $w$)$  =\frac{ p_{\mathcal{R}}^{a} (w) }{ p_{\mathcal{R}}^{a} (w) + p_{\mathcal{R}}^{r} (w) } $
	and 
	P($\mathcal{R}$ rejects $w$)$ =\frac{ p_{\mathcal{R}}^{r} (w) }{ p_{\mathcal{R}}^{a} (w) + p_{\mathcal{R}}^{r} (w) } $.
\end{lemma}
\begin{proof}
	\begin{eqnarray*}
		\mbox{P(}\mathcal{R}\mbox{ accepts }w) & = &
        \sum_{i=0}^{\infty}\left(1-p_{\mathcal{R}}^{a}(w)-p_{\mathcal{R}}^{r}(w) \right)^{i}
        p_{\mathcal{R}}^{a}(w)\\
		& = & p_{\mathcal{R}}^{a}(w) \left(
		\dfrac{1}{1-(1-p_{\mathcal{R}}^{a}(w)-p_{\mathcal{R}}^{r}(w))} \right) \\
		& = &
		\dfrac{p_{\mathcal{R}}^{a}(w)}{p_{\mathcal{R}}^{a}(w)+p_{\mathcal{R}}^{r}(w)}
	\end{eqnarray*}
	P($\mathcal{R}$ rejects $w$) is calculated in the same way.
\end{proof}

% sssssssssssssssssssssssssssssssssssssssssssssssssssssssssssssssssssssssssssssssss %
% sssssssssssssssssssssssssssssssssssssssssssssssssssssssssssssssssssssssssssssssss %
\section{Turing machines with postselection} \label{sec:Posdefs}
% sssssssssssssssssssssssssssssssssssssssssssssssssssssssssssssssssssssssssssssssss %
% sssssssssssssssssssssssssssssssssssssssssssssssssssssssssssssssssssssssssssssssss %

We are now ready to present our model of computation with postselection.

Turing machines with postselection are defined  such that the overall state set is partitioned into four subsets, namely, the sets of continuing ($ Q_{c} $), postselection accept ($ Q_{pa} $), postselection reject ($ Q_{pr} $), and nonpostselection halting ($ Q_{nh} $) states. Correspondingly, we set $ \Delta = \{ c,pa,pr,nh\} $.
For these machines, $ Q_{h} = Q_{pa} \cup Q_{pr} \cup Q_{nh} $.
The computation is terminated whenever the machine enters a state in $ Q_{h} $. For every possible input, a Turing machine with postselection always halts in a state in $Q_{pa} \cup Q_{pr}$ with nonzero probability.

Probabilistic and quantum real-time finite automata with postselection (PFAPs and QFAPs) are obtained by specializing the TM version in the manner described in Section \ref{sec:Prel}; see Appendix \ref{app:postfa} for detailed definitions.

Let $ p^{a}_{\mathcal{P}} (w) $ (resp. $ p^{r}_{\mathcal{P}} (w) $) be the probability that a TM with postselection named $ \mathcal{P} $ reaches a state in $ Q_{pa} $ (resp. $ Q_{pr} $) when run on an input  $w$.\footnote{Note that we are using notation identical to that introduced in the discussion for machines with restart for these probabilities; the reason will
be evident shortly.} For each such $ w \in \Sigma^{*} $,
the overall acceptance and rejection probabilities of $w$ are obtained by normalization, and are given by
\begin{equation}
\label{eq:postacc}
       \mbox{P(}\mathcal{P}\mbox{ accepts }w) =
\dfrac{p_{\mathcal{P}}^{a}(w)}{p_{\mathcal{P}}^{a}(w)+p_{\mathcal{P}}^{r}(w)},
\end{equation}
and
\begin{equation}
\label{eq:postrej}
       \mbox{P(}\mathcal{P}\mbox{ rejects }w) =
\dfrac{p_{\mathcal{P}}^{r}(w)}{p_{\mathcal{P}}^{a}(w)+p_{\mathcal{P}}^{r}(w)}.
\end{equation}
``Postselection" is the name given to this process, where  any computational path ending with a transition to a state in $Q_{nh} $ is simply discarded, and only the ones ending with a state in  $Q_{pa} \cup Q_{pr}$ are ``selected".

For every language class $\mathbf{C}$ defined using a resource-bounded probabilistic or quantum machine model under a particular error regime, we define the class $\mathsf{Post}\mathbf{C}$ of the languages recognized under the same error regime by the corresponding type of machines with postselection.

Some other classes that will be studied are presented in Table \ref{tbl:postrt-class}. (We need to name these  separately, since it is easy to see that the standard versions of all these classes correspond to the same standard class, namely, the regular languages. The postselected versions are not identical, as will be apparent in the subsequent sections.)

\begin{table}[h!]
\caption{Classes of languages recognized by real-time constant-memory machines with postselection under different error regimes}
\centering
\fbox{
\begin{minipage}{0.95\textwidth}
\[
	\begin{array}{lccc}
		\mbox{\underline{machine type}} & \mbox{\underline{two-sided bounded-error}} & 
			\mbox{\underline{one-sided bounded-error}} &  \mbox{\underline{error-free}}
		\\
		\mbox{PFAP} & \mathsf{PostBS} & \mathsf{PostRS} & \mathsf{PostES}
		\\
		\mbox{QFAP} & \mathsf{PostBQAL} & \mathsf{PostRQAL} & \mathsf{PostEQAL}
				\\
		\mbox{real-time PTM} & \mathsf{Post}\textsc{rt}\mathsf{BPSPACE}(1) & \mathsf{Post}\textsc{rt}\mathsf{RSPACE}(1) 
			& \mathsf{Post}\textsc{rt}\mathsf{EPSPACE}(1) 
		\\
		\mbox{real-time QTM} & \mathsf{Post}\textsc{rt}\mathsf{BQSPACE}(1) & \mathsf{Post}\textsc{rt}\mathsf{RQSPACE}(1) 
			& \mathsf{Post}\textsc{rt}\mathsf{EQSPACE}(1)
	\end{array}
\]
\end{minipage}
}	
\label{tbl:postrt-class}
\end{table}

Since none of the results on real-time constant-space machines in this paper are sensitive to the existence of uncomputable numbers among the transition probabilities of the program, we will mostly use the shorter class names in the top two rows of Table \ref{tbl:postrt-class} in the subsequent sections, with the implication that the same relationship is valid among the corresponding classes in the bottom two rows.

Trivially, every postselected class contains its standard version with the same resource bounds, and we are interested in finding out whether the inclusion is proper or not.

Although we will not have much to say about classes defined solely in terms of time bounds, the reader will note that Aaronson's $ \mathsf{PostBQP}$ \cite{Aa05} has an equivalent definition in terms of our model, as the class of languages recognized with bounded error by polynomial-time QTMs with postselection.\footnote{A detailed treatment of the equivalence between the TM and circuit models of quantum computation can be found in \cite{Ya93}.}

% sssssssssssssssssssssssssssssssssssssssssssssssssssssssssssssssssssssssssssssssss %
% sssssssssssssssssssssssssssssssssssssssssssssssssssssssssssssssssssssssssssssssss %
\section{The power of postselection}\label{sec:powerofpost}
% sssssssssssssssssssssssssssssssssssssssssssssssssssssssssssssssssssssssssssssssss %
% sssssssssssssssssssssssssssssssssssssssssssssssssssssssssssssssssssssssssssssssss %

It is evident from the similarity of the statement of Lemma
\ref{lem:overall-acc-rej} and Equations \ref{eq:postacc}
and \ref{eq:postrej} that there is a close relationship between
machines with restart and those with postselection. This is set out in
the following theorem.

\begin{theorem}
\label{thm:posres}

For any time bound $t$ and space bound $s$, the class of languages
recognized by
$ t $-time and $ s $-space PTMs (resp. QTMs) with postselection is
identical to the class of languages recognized by
$ t $-time PTMs (resp. QTMs) with restart using space $ s $. The same
equality is also valid for the real-time, in particular, finite
memory, versions of these models.

\end{theorem}
\begin{proof}
Given a (probabilistic or quantum) machine with postselection called
$\mathcal{P}$, we can construct a corresponding machine with restart
called  $\mathcal{R}$, which is identical to $\mathcal{P}$, except that all nonpostselection halting states of $\mathcal{P}$ are designated as
restart states in $\mathcal{R}$. $\mathcal{R}$'s accept and reject states correspond precisely to 
$\mathcal{P}$'s postselection accept and reject states, respectively.

Given a machine with restart  $\mathcal{R}$,
we construct a corresponding machine with postselection  $\mathcal{P}$
by starting with an exact copy of $\mathcal{R}$, designating the  old accept and reject states  as  the
postselection accept and reject states of $\mathcal{P}$, respectively,
and converting the restart states to 
nonpostselection halting states.

       By Lemma \ref{lem:overall-acc-rej} and Equations \ref{eq:postacc} and
\ref{eq:postrej}, the machines before and after these conversions
recognize the same language, with the same error bound.
\end{proof}

Theorem \ref{thm:posres} will be useful in our analyses of the
language classes corresponding to machines with postselection. One immediate corollary is that the postselected versions of classes  that are defined solely in terms of space bounds (for machines with two-way input heads) are equal to the corresponding classes for standard machines.

So the additional power brought by the capability of postselection is nil for space-bounded machines, and is probably difficult to prove for time-bounded machines. We therefore consider machines operating under simultaneous time-space bounds, and  are now able to demonstrate that postselection increases the
recognition power of both probabilistic and quantum computers.
 For a given string $ w $, let $ |w|_{\sigma} $ denote the number of
occurrences of symbol $ \sigma $ in $ w $.

\begin{theorem}
       \label{thm:postpbeatsp}
        PTMs with postselection that use $o(\log \log n)$ space and
        that have polynomial expected runtime are strictly more powerful
than their standard versions.
\end{theorem}
\begin{proof}
       As proven by Dwork and Stockmeyer, standard polynomial-time PTMs that
use $o(\log \log n)$ space recognize precisely the regular languages
\cite{DS90}.    The nonregular language $ L_{eq} = \{ w \in \{a,b\}^{*}
\mid |w|_{a} = |w|_{b} \} $ can be recognized by
       a real-time PTM with restart using $O(1)$ space \cite{YS10B}. By
Theorem \ref{thm:posres},
       $ L_{eq} \in \mathsf{Post}\textsc{rt}\mathsf{BPSPACE}(1)$.
       We conclude that the $\mathsf{BPTISP}$ classes are properly contained in the $\mathsf{PostBPTISP}$ classes for these time-space bounds.
\end{proof}

If the PTMs in question are further restricted so that their input
tape heads can never move left, we do not need an expected value to
exist for the runtime:
\begin{theorem}
       \label{thm:one-way-postpbeatsp}
       One-way PTMs with postselection that use $o(\log \log n)$ space
       are strictly more powerful than their standard versions.
\end{theorem}
\begin{proof}
        One-way PTMs that use $o(\log \log n)$ space are known to recognize precisely
       the regular languages \cite{Fr85}. The remainder follows the proof of
Theorem \ref{thm:postpbeatsp}.
\end{proof}

A similar advantage can also be demonstrated for QTMs, albeit under
more severe bounds. $w ^{r}$ denotes the reverse of string $ w $.
\begin{theorem}
       \label{thm:posqbeatsq}
       $ \textsc{rt}\mathsf{BQSPACE(1)} \subsetneq
\mathsf{Post}\textsc{rt}\mathsf{BQSPACE(1)} \subseteq \mathsf{PostBQAL}
$.
\end{theorem}
\begin{proof}
       It is known \cite{KW97,Bo03,Je07} that
\textsc{rt}$\mathsf{BQSPACE(1)}=\mathsf{REG}$.
       The nonregular language $ L_{pal} = \{w \in \{a,b\}^{*} \mid w = w
^{r} \} $ can be recognized by
       a real-time QTM with restart using $O(1)$ space \cite{YS10B}. By
Theorem \ref{thm:posres},
       $ L_{pal} \in \mathsf{Post}\textsc{rt}\mathsf{BQSPACE(1)} $, leading
us to conclude that QTMs with postselection outperform their standard
counterparts under these time-space bounds.
\end{proof}

We are also able to show that quantum postselection machines
outperform their classical counterparts. This follows easily from the well-known relation
\begin{equation}\label{eq:aw02}
\mathsf{BPSPACE}(s) \subsetneq \mathsf{BQSPACE}(s)
\end{equation}
for any space bound $s=o(\log n)$, ($ L_{pal} \in \mathsf{BQSPACE(1)} $ \cite{AW02}, but $ L_{pal} \notin \mathsf{BPSPACE}(s) $ for such $s$ \cite{FK94}) since we now know that these classes equal their postselected versions. We wish to find a setup where an
advantage of a standard quantum model over its classical counterpart
has not been shown, and demonstrate the superiority of postselected
quantum over postselected probabilistic  in that context. The real-time
restriction is again seen to be useful in this regard. It is presently
not known whether $\textsc{rt}\mathsf{BPSPACE}(s)=
\textsc{rt}\mathsf{BQSPACE}(s)$ or not for any space bound
$s=\omega(1)$. However, the fact that we have a real-time QFA with postselection for recognizing  $ L_{pal} $ (Theorem
\ref{thm:posqbeatsq}), combined with the argument above for Equation \ref{eq:aw02}, lead us to conclude that
\[\mathsf{Post}\textsc{rt}\mathsf{BPSPACE}(s)\subsetneq
\mathsf{Post}\textsc{rt}\mathsf{BQSPACE}(s) \mbox{ for all } s=o(\log n).\]

Aaronson \cite{Aa05} showed that any $\mathsf{PostBQP}$ computation
can be repeated a polynomial number of times to reduce the error
probability to $2^{-p(n)}$ for any desired polynomial $p$, and this
applies easily to our generalized classes, with one exception: Since
it is not allowed to increase the runtime of real-time machines, error reduction for the
$\mathsf{Post}\textsc{rt}\mathsf{BPSPACE}$ and
$\mathsf{Post}\textsc{rt}\mathsf{BQSPACE}$ classes has to be performed
by repeating the original computation not sequentially, but
parallelly, essentially by increasing the size of the program. We know
how to do this parallelization only for $O(1)$-space machines; see
Appendix \ref{app:postfa-error}.

\begin{theorem}
       \label{thm:post-closure}      
    $\mathsf{PostBS}$, $
\mathsf{PostBQAL}$, and the classes $\mathsf{PostBPTIME}(t)$ and 
       $\mathsf{PostBQTIME}(t)$ (for every time bound $t$)
       are closed under complementation, union, and intersection.
Furthermore the classes
       $\mathsf{Post}\textsc{rt}\mathsf{BPSPACE}(s)$ and $
\mathsf{Post}\textsc{rt}\mathsf{BQSPACE}(s)$
       are closed under complementation for any space bound $s$.
\end{theorem}
\begin{proof}
As above, Aaronson's proof of $\mathsf{PostBQP}$'s closure properties
can be adapted easily for our classes,
except for real-time machines using non-constant space.
(See Appendix \ref{app:postfa-closure}.)
\end{proof}

\begin{theorem}
       \label{thm:PostQ-subset-Q}
       For any time bound $t$ and space bound $s$,
       \begin{itemize}
               \item  $\mathsf{PostBPTISP}(t,s) \subseteq \mathsf{PrTISP}(t,s)$, and,
               \item  $\mathsf{PostBQTISP}(t,s) \subseteq \mathsf{PrQTISP}(t,s)$.
       \end{itemize}
       Furthermore,   $ \mathsf{PostBS} \subseteq
\mathsf{S} $, and $ \mathsf{PostBQAL} \subseteq
\mathsf{QAL} $.
\end{theorem}
\begin{proof}
       A given machine $\mathcal{P}$ with postselection can be converted to a machine in
the corresponding standard model (without
        postselection) as follows:
       \begin{itemize}
               \item All transitions to nonpostselection halting states of $\mathcal{P}$ at the end of computation are replaced by two equiprobable transitions to accept and reject states.
               \item All  postselection accept states of $\mathcal{P}$ are designated as accept states in the new machine.
       \end{itemize}
       Therefore, only the strings which are members of the original machine's
language are accepted with probability
       exceeding $ \frac{1}{2} $ by the new machine.
\end{proof}

For probabilistic machines, the result above can be strengthened so that bounded-error
postselected probabilistic  space is shown to be properly contained
in standard unbounded-error probabilistic  space,
without requiring simultaneous time bounds, for all sublogarithmic space bounds. We recall the ineffectiveness of postselection for space bounded machines, and use the following fact:

\begin{fact}
 For any space bound $ s \in o(\log n) $, \[ \mathsf{BPSPACE}(s)
\subsetneq  \mathsf{PrSPACE}(s). \]
 Moreover,  $ \mathsf{BPSPACE(1)} \subsetneq
\textsc{rt}\mathsf{PrSPACE(1)} $.
\end{fact}
\begin{proof}
       As mentioned above, $ L_{pal} \notin \mathsf{BPSPACE}(s) $ for any
such $s$, but $ L_{pal} \in \mathsf{PrSPACE}(1)$.
       For the second relation, we also use the fact that $ \mathsf{PrSPACE(1)} =
\textsc{rt}\mathsf{PrSPACE(1)} $ \cite{Ka91}.
\end{proof}

Recalling Aaronson's celebrated result stating the equality of
bounded-error postselected quantum polynomial time to standard
unbounded-error probabilistic polynomial time, we ask whether the same
relationship holds for real-time automata. The answer turns out to be
negative.

\begin{theorem}\label{thm:aarnoo}
      $ \mathsf{PostBQAL} \subsetneq \mathsf{S}  $.
\end{theorem}
\begin{proof}
       By Theorem \ref{thm:PostQ-subset-Q}, we have
       $ \mathsf{PostBQAL} \subseteq \mathsf{QAL} $.
       It is known \cite{YS11A} that
       $ \mathsf{QAL}= \mathsf{S} $, and
       that $ \mathsf{S} $ is not
closed under union and intersection
       \cite{Fl72,Fl74,La74,Tu82}. We conclude by Theorem \ref{thm:post-closure} that
       the containment must be proper.
\end{proof}

For instance, for any triple of integers $u$, $v$, $w$, where
$0<u<v<w$, the languages $L_{1}= \{a^{m}b^{k}c^{n} | m^{u}>k^{v}>0\}$
and $L_{2}= \{a^{m}b^{k}c^{n} | k^{v}>n^{w}>0\}$ are in $ \textsc{rt}\mathsf{PrSPACE(1)} $,
whereas $L_{1} \cup L_{2}$ is not in $ \mathsf{S} $ \cite{Tu82}. It must
therefore be the case that at least one of $L_{1}$ and $L_{2}$ is not
in $ 
\mathsf{PostBQAL} $.

% sssssssssssssssssssssssssssssssssssssssssssssssssssssssssssssssssssssssssssssssss %
% sssssssssssssssssssssssssssssssssssssssssssssssssssssssssssssssssssssssssssssssss %
\section{Machines with postselection under other error regimes}\label{sec:errorfree}
% sssssssssssssssssssssssssssssssssssssssssssssssssssssssssssssssssssssssssssssssss %
% sssssssssssssssssssssssssssssssssssssssssssssssssssssssssssssssssssssssssssssssss %
Whether the permission to commit two-sided, rather than one-sided, error enlarges the class of  languages recognized by probabilistic programs is an open question for a wide range of resource bounds. For standard quantum programs, we do not even know if error-free computation  is equivalent to deterministic computation or not \cite{ADH97}. In this section, we examine these issues for machines with postselection.

We start with a characterization of the error-free classes. 
\begin{theorem}\label{thm:exactchar}
 For every time bound $t$,
\begin{itemize}
\item $\mathsf{PostEPTIME}(t) = \mathsf{NTIME}(t)\cap \mathsf{coNTIME}(t)$,

\item $\mathsf{PostEQTIME}(t) = \mathsf{NQTIME}(t)\cap \mathsf{coNQTIME}(t)$.

\end{itemize}
\end{theorem}
\begin{proof}
We can convert a  machine with postselection $\mathcal{M}$ that recognizes its language $L$ with no error to an equivalent nondeterministic machine $\mathcal{M'}$ which operates within the same time and space bounds by simply designating all nonpostselection halting states of $\mathcal{M}$ as reject states (in addition to any more reject states inherited from the original definition) in $\mathcal{M'}$. If  in addition to this transformation, we also switch the designations of the original postselection halting states, we obtain a nondeterministic machine recognizing the complement of $L$.

For the inclusion in the other direction, let $ \mathcal{M}_{1} $ and $ \mathcal{M}_{2} $ be two probabilistic (resp. quantum) TMs recognizing a language $ L $, 
and its complement, respectively, with cutpoint zero in time $t$. We build a $ O(t) $-time  PTM (resp. QTM)
with postselection named $ \mathcal{R} $ that 
runs $ \mathcal{M}_{1} $ and $ \mathcal{M}_{2} $ separately on its input. 
If $ \mathcal{M}_{1} $ accepts, $ \mathcal{R} $ accepts. 
Otherwise $ \mathcal{R} $ runs $ \mathcal{M}_{2} $ on the input. 
If $ \mathcal{M}_{2} $ accepts, $ \mathcal{R} $ rejects. 
Otherwise, $ \mathcal{R} $ halts in a nonpostselection state. It is easy to see that  $ \mathcal{R} $  makes no error, and halts with probability 1.
\end{proof}

This construction can be modified easily to apply to the real-time finite automata cases as well, and  we can  conclude

\begin{corollary}\label{corollary:exactfa}
	$ \mathsf{PostES}=\mathsf{REG}  $, and $ \mathsf{PostEQAL} = \mathsf{NQAL} \cap \mathsf{coNQAL} $.
\end{corollary}

(Note that it is still open  whether $\mathsf{NQAL} \cap \mathsf{coNQAL}$ 
contains a nonregular language or not, though we do know that $\mathsf{NQAL} \neq \mathsf{coNQAL}  $ \cite{YS10A}.)

Theorem \ref{thm:exactchar} implies, for instance, that since the decision version of the integer factorization problem is in $\mathsf{NP}\cap \mathsf{coNP}$,  there exists an error-free polynomial-time probabilistic TM with postselection which recognizes that language, whereas the only known standard (quantum) algorithm with worst-case polynomial-time for this problem \cite{Sh97} commits bounded error.\footnote{After the completion of this paper, it came to our attention that Brun
and Wilde \cite{BW11} made similar remarks in the context of
computation using postselected closed timelike curves.}

The reader will note that the transformations in the proof of Theorem \ref{thm:exactchar} do not increase the space usage of the machines in question. Using the equivalence of space-bounded machines with and without postselection, we can therefore view those constructions as providing an alternative proof for the facts
\[\mathsf{EPSPACE}(s) = \mathsf{NSPACE}(s)\cap \mathsf{coNSPACE}(s)\]
and
\[\mathsf{EQSPACE}(s) =\mathsf{NQSPACE}(s)\cap \mathsf{coNQSPACE}(s)\]
for all  $ s $.
Making use of the Immerman-Szelepcs\'{e}nyi theorem, we conclude 
\[\mathsf{EPSPACE}(s) =\mathsf{NSPACE}(s)= \mathsf{RSPACE}(s)\]
for all $s=\Omega(\log n)$, 
meaning that error-free probabilistic machines (with or without postselection) are equivalent those operating with one-sided error under these space bounds.\footnote{Note that $\mathsf{EPSPACE}(s) = \mathsf{NSPACE}(s)$ also follows for all $s=\Omega(\log n)$ by a modification of Gill's proof \cite{Gi77} of $\mathsf{RSPACE}(s) = \mathsf{NSPACE}(s)$ to take the Immerman-Szelepcs\'{e}nyi theorem into account, without the need to talk about postselection. We suspect that this is a well-known fact, but we have not seen it stated anywhere.} 

The error-free and one-sided-error modes of computation are also equivalent for PFAPs, since every language in $\mathsf{PostRS} $ obviously has a two-way PFA recognizing it with one-sided error, and those machines are equivalent to deterministic finite automata:
\begin{equation}\label{eq:postrspostes}
\mathsf{PostRS}=\mathsf{REG}= \mathsf{PostES}.
\end{equation}

Things change in the quantum case.

\begin{theorem}\label{thm:onesidedbeatsexact}
 $ \mathsf{PostEQAL} \subsetneq \mathsf{PostRQAL}\subseteq \mathsf{NQAL}$.
\end{theorem}
\begin{proof}
	The complement of $L_{pal}$ can be recognized with one-sided bounded error by a real-time QFA with restart \cite{YS10B}, and is therefore in  $\mathsf{PostRQAL}$ by Theorem \ref{thm:posres}. But the same language is known to be outside $\mathsf{NQAL} \cap \mathsf{coNQAL} $ \cite{YS10A}. The proper containment follows using Corollary \ref{corollary:exactfa}. The second subset relationship is given by a simplification of the proof of Theorem \ref{thm:exactchar}.
\end{proof}

For greater $t$, the same simplification to the proof of Theorem \ref{thm:exactchar} mentioned
above  yields 
\[ \mathsf{PostRTIME}(t) \subseteq \mathsf{NTIME}(t)), \]
and 
\[ \mathsf{PostRQTIME}(t) \subseteq \mathsf{NQTIME}(t), \]
which is not enough to say whether one-sided error is useful in this range.
 
An open problem regarding classical nondeterministic space is whether $\mathsf{NL}=\mathsf{RL}$ or not. (It is known that $\mathsf{RSPACE}(s)=\mathsf{NSPACE}(s)$ for all $s$, but the proof of this fact \cite{Gi77} involves the construction of a randomized machine that runs in time that is double exponential in terms of $s$. Recall from Section \ref{sec:Prel} that $\mathsf{RL}\equiv\mathsf{RTISP}(poly(n),\log n)$.) We prove the equality in the case where postselection is allowed:

\begin{theorem}\label{thm:nspace}
 For every space constructible function $s$,
$\mathsf{NSPACE}(s) =  \mathsf{PostRTISP}(2^{O(s)},s)$.
\end{theorem}
\begin{proof}
We first note that every nondeterministic machine which uses space $ s(n)=\Omega(\log n)$ must have certificates of length at most $l=2^{O(s(n))}$ for every accepted string of length $n$, and every such machine where  $s$ is space-constructible   can incorporate a counter that helps cut the  execution off if the runtime exceeds $l$, so $\mathsf{NSPACE}(s)=\mathsf{NTISP}(2^{O(s)},s)$.

By Theorem \ref{thm:posres}, $\mathsf{PostRTISP}(2^{O(s)},s)$ is contained in $\mathsf{RSPACE}(s)$, which, as already mentioned, equals $\mathsf{NSPACE}(s)$ for every $s$. In the other direction, for any nondeterministic TM $\mathcal{N}$ that uses space $s$, we build a PTM with postselection $\mathcal{A}$ as follows: Upon reading an input of length $n$, $\mathcal{A}$ first computes the maximum length $l$, mentioned in the previous paragraph, of the shortest certificate that $\mathcal{N}$ could possibly use for such an input. $\mathcal{A}$ then rejects the input with probability $2^{-2l}$. With the remaining probability, $\mathcal{A}$ randomly guesses a certificate of length at most $l$ for the input, checking it in at most $2^{O(s)}$ time in space $s$, accepting if it is a valid certificate, and halting in a nonpostselection state otherwise. $\mathcal{A}$ rejects nonmembers of the language of $\mathcal{N}$ with probability 1, and accepts members with probability greater than $\frac{2}{3}$. 
\end{proof}

Finally, we would like to see whether the permission to commit two-sided, rather than one-sided, error gives any advantage to machines with postselection. For any language $L$, let $\overline{L}$ denote the complement of $L$.

\begin{theorem}\label{thm:2posqbeats1q}
	$ \mathsf{PostRQAL}  \subsetneq \mathsf{PostBQAL} $.
\end{theorem}
\begin{proof}	
	The language $ L_{eq\overline{eq}}  = \{ aw_{1} \cup bw_{2} \mid w_{1} \in L_{eq}, w_{2} \in \overline{L_{eq}} \} $ is not a member of $ \mathsf{NQAL} $ \cite{YS10A}, and therefore also not in  $ \mathsf{PostRQAL} $, by Theorem \ref{thm:onesidedbeatsexact}. We will construct a PFAP $\mathcal{P}$ that recognizes $ L_{eq\overline{eq}}$ with two-sided bounded error.
	
	Since $L_{eq} \in \mathsf{PostBS}$ (Theorem \ref{thm:postpbeatsp}), its complement, $\overline{L_{eq}}$, is also in $\mathsf{PostBS}$ (Theorem \ref{thm:post-closure}). Let the corresponding PFAPs be called $\mathcal{M}_1$ and $\mathcal{M}_2$, respectively. $\mathcal{P}$ handles inputs shorter than two symbols deterministically. For longer inputs, it passes control to $\mathcal{M}_1$ or $\mathcal{M}_2$, depending on whether the first symbol is an $a$ or a $b$.
	
	Since $ L_{eq\overline{eq}} \in \mathsf{PostBS} \subsetneq \mathsf{PostBQAL}$, the statement has been proven.
\end{proof}

In the probabilistic case, the superiority of finite automata with two-sided error follows from Theorem \ref{thm:postpbeatsp} and Equation \ref{eq:postrspostes}, but we can do better than this: No  nondeterministic TM (and therefore, no PTM
with postselection that commits one-sided bounded error) using $o(\log n)$ space can recognize a nonregular deterministic context-free language \cite{AGM92}. Since $L_{eq}$ (Theorem \ref{thm:postpbeatsp}) is such a language, for which there exists a bounded-error constant-space PTM, we conclude that the advantage given by the flexibility to  err both ways extends to all sublogarithmic bounds for two-way classical machines (with or without postselection).

% SSSSSSSSSSSSSSSSSSSSSSSSSSSSSSSSSSSSSSSSSSSSSSSSSSSSSSSSSSSSSSSSSSSSSSSSSSSSSSSS %
% SSSSSSSSSSSSSSSSSSSSSSSSSSSSSSSSSSSSSSSSSSSSSSSSSSSSSSSSSSSSSSSSSSSSSSSSSSSSSSSS %
% SSSSSSSSSSSSSSSSSSSSSSSSSSSSSSSSSSSSSSSSSSSSSSSSSSSSSSSSSSSSSSSSSSSSSSSSSSSSSSSS %
\section{Riga quantum finite automata with postselection} \label{sec:LPostFA}
% SSSSSSSSSSSSSSSSSSSSSSSSSSSSSSSSSSSSSSSSSSSSSSSSSSSSSSSSSSSSSSSSSSSSSSSSSSSSSSSS %
% SSSSSSSSSSSSSSSSSSSSSSSSSSSSSSSSSSSSSSSSSSSSSSSSSSSSSSSSSSSSSSSSSSSSSSSSSSSSSSSS %
% SSSSSSSSSSSSSSSSSSSSSSSSSSSSSSSSSSSSSSSSSSSSSSSSSSSSSSSSSSSSSSSSSSSSSSSSSSSSSSSS %

The first automaton-based model of computation with postselection was presented by L\={a}ce, Scegulnaja-Dubrovska and Freivalds, \cite{LSF09} who were interested exclusively in quantum finite automata, that is, real-time, constant-space QTMs. The main difference between L\={a}ce \textit{et al.}'s way of modeling postselection and our approach is that the transitions of a \textit{Riga QFA with postselection} (RQFAP), as we name their model, are not assumed to lead the machine to at least one postselection state with nonzero probability. RQFAPs have the additional unrealistic capability of detecting if the total probability of postselection states is zero at the end of the processing of the input, and accumulating all probability in a single output in such a case.

Although the motivation for this feature is not explained in
\cite{LSF09,SLF10}, such an approach may be seen as an attempt to compensate for some fundamental weaknesses of finite automata.
In many computational models with bigger space bounds, 
one can modify a machine employing the Riga approach without changing the recognized language 
so that the postselection state set will have nonzero probability for any input string. 
This is achieved by just creating some
computational paths that end up in the postselection set with
sufficiently small probabilities so that
their inclusion does not change the acceptance probabilities of strings that lead the original machine to the postselection set significantly. These paths can be used to accept or to reject
the input as desired whenever there is zero probability of observing the other postselection states.
Unfortunately, we do not know how to implement this construction
in quantum finite automata with arbitrary amplitudes, so we prefer our model, in which the only nonstandard capability conferred to the machines is postselection, to the Riga version.

In this section, we will  prove that the Riga model is not equivalent to ours in 	
	computational power. More results on the properties of these machines can be found in Appendix \ref{app:lpostfa}.

We will consider Riga finite automata with postselection\footnote{The original definitions of RQFAPs in \cite{LSF09} are based on weaker QFA variants. Replacing those with the machines of Appendix \ref{app:postfa} does not change the model 
\cite{YS11B}). The classical probabilistic versions of these machines are studied here for the first time.} as finite-state machines of the type introduced in Section \ref{sec:Posdefs} (also see Appendix \ref{app:postfa}) with an additional component $ \chi \in \{A,R\} $, such that
whenever the postselection probability is zero for a given input string $ w \in \Sigma^{*} $,
\begin{itemize}
       \item $ w $ is accepted with probability 1 if $ \chi = A $,
       \item $ w $ is rejected with probability 1 if $ \chi = R $.
\end{itemize}
The related language classes are named by prefixing the letter $\mathsf{\frak{R}}$ to the corresponding class name from Table \ref{tbl:postrt-class}.

In the classical case, Riga machines are equal in power to ours, that is,
$ \mathsf{\frak{R}PostBS} $ = $ \mathsf{PostBS} $ and $ \mathsf{PostES} $=$ \mathsf{\frak{R}PostES} $ 
(see Appendix \ref{app:lpostfa} for a proof).

 Recall from Section \ref{sec:powerofpost} that 
 $ \mathsf{PostEQAL}=\mathsf{NQAL} \cap \mathsf{coNQAL}$. 
 We will now show that the corresponding class for RQFAPs is larger.

\begin{theorem}\label{thm:nuncon}
	$ \mathsf{\frak{R}PostEQAL} = \mathsf{NQAL} \cup \mathsf{coNQAL} $.
\end{theorem}
\begin{proof}
	For $ L \in \mathsf{NQAL} $, designate the accepting states of the QFA recognizing 
	$ L $ with cutpoint zero as postselection accepting states with $ \chi = R $. 
	(There are no postselection reject states.)
	
	For $ L \in \mathsf{coNQAL} $, designate the accepting states of the QFA recognizing 
	the complement of $ L $ with cutpoint zero
	as postselection rejecting states with $ \chi = A $. 
	(There are no postselection accept states.)
	
	Finally, let $ L $ be a member of $ \mathsf{\frak{R}PostEQAL}  $ and $ \mathcal{M} $ 
	be a RQFAP recognizing $ L $ with zero error.
	If $ \chi=R $, we have that, for all $ w \in L $, $ p^{a}_{\mathcal{M}}(w) $ is nonzero,
	 and $ p^{r}_{\mathcal{M}}(w) = 0 $, and for all $w \notin L$,  $ p^{a}_{\mathcal{M}}(w)=0 $.
	Thus, we can design a real-time QFA recognizing $ L $ with cutpoint zero.
	If $ \chi=A $, we can similarly design a real-time
QFA recognizing the complement of $L$ with cutpoint
zero.
\end{proof}

By using the fact\footnote{$ L_{pal} $ was proven to be in $ \mathsf{\frak{R}PostEQAL} $ for the first time in \cite{LSF09}.} that $ L_{pal} \in $ $ \mathsf{coNQAL} $ $ \setminus $ $ \mathsf{NQAL} $ \cite{YS10A},
we can state that RQFAPs are strictly more powerful than our version of 
real-time QFAs with postselection, at least in the error-free mode:

\begin{corollary}
	$ \mathsf{PostEQAL}  $ $ \subsetneq $ $ \mathsf{\mathsf{\frak{R}PostEQAL}} $.
\end{corollary}

In the bounded-error case, it is not known whether RQFAPs can outperform QFAPs or not.

% SSSSSSSSSSSSSSSSSSSSSSSSSSSSSSSSSSSSSSSSSSSSSSSSSSSSSSSSSSSSSSSSSSSSSSSSSSSSSSSS %
% SSSSSSSSSSSSSSSSSSSSSSSSSSSSSSSSSSSSSSSSSSSSSSSSSSSSSSSSSSSSSSSSSSSSSSSSSSSSSSSS %
% SSSSSSSSSSSSSSSSSSSSSSSSSSSSSSSSSSSSSSSSSSSSSSSSSSSSSSSSSSSSSSSSSSSSSSSSSSSSSSSS %
\section{Concluding remarks} \label{sec:ConcludingRemarks}
% SSSSSSSSSSSSSSSSSSSSSSSSSSSSSSSSSSSSSSSSSSSSSSSSSSSSSSSSSSSSSSSSSSSSSSSSSSSSSSSS %
% SSSSSSSSSSSSSSSSSSSSSSSSSSSSSSSSSSSSSSSSSSSSSSSSSSSSSSSSSSSSSSSSSSSSSSSSSSSSSSSS %
% SSSSSSSSSSSSSSSSSSSSSSSSSSSSSSSSSSSSSSSSSSSSSSSSSSSSSSSSSSSSSSSSSSSSSSSSSSSSSSSS %

Figure  \ref{fig:rt-rel} summarizes our results on constant-memory computers. Dotted arrows indicate subset 
relationships, and unbroken arrows represent the cases where it is known that the inclusion is proper. Note that the subscript $\mathbb{R}$ appearing in the TM-related class names indicates that the transition probabilities/amplitudes are allowed to be unrestricted real numbers, making them correspond to two-way finite automata. In comparisons involving probabilistic computers, we were able to demonstrate the superiority results in the figure for some larger time and space bounds.

\begin{figure}[h!]
	\begin{center}
	\fbox{ \footnotesize
	\begin{minipage}{0.95\textwidth}
		\begin{center}
		~~\\~~\\
		\ifx\JPicScale\undefined\def\JPicScale{1}\fi
		\unitlength \JPicScale mm
		\begin{picture}(150,105)(0,0)
			% LEVEL 1
			\put(35,4){$ \mathsf{REG = PostES = PostRS = \frak{R}PostES = \frak{R}PostRS } $}
			\put(35,0){$ \mathsf{ = EPSPACE_{\mathbb{R}}(1) = RSPACE_{\mathbb{R}}(1) = NSPACE_{\mathbb{R}}(1) } $}

			% LEVEL 2
			\put(35,15){$ \mathsf{NQAL \cap coNQAL} $}
			\put(35,19){$ \mathsf{PostEQAL} = $}
						
			% CONNECT 1 TO 2			
			\multiput(50,7)(0,1){7}{$ . $} \put(50.3,13.5){\vector(0,1){1.5}}
			
			% LEVEL 3			
			\put(0,30){$ \mathsf{PostBS = \frak{R}PostBS} $}			
			\put(100,30){$ \mathsf{NQSPACE_{\mathbb{R}}(1) \cap coNQSPACE_{\mathbb{R}}(1) } $}
			\put(100,34){$ \mathsf{ EQSPACE_{\mathbb{R}}(1) = } $}
			\put(45,30){$ \mathsf{PostRQAL} $}
			\put(70,30){$ \mathsf{NQAL} $}

			% CONNECT 1 TO 3
			\put(35,7){\vector(-1,2){11}}
			\multiput(103,7)(1,3){8}{$ . $} \put(110,26){\vector(1,3){1}}
			
			% CONNECT 2 TO 3
			\multiput(55,19)(3,1){15}{$ . $} \put(98,33.4){\vector(3,1){1.5}}
			\put(50,22){\vector(0,1){8}}			
			
			% LEVEL 4						
			\put(65,45){$ \mathsf{NQAL \cup coNQAL} $}
			\put(65,49){$ \mathsf{ \frak{R}PostEQAL = } $}
			\put(100,45){$ \mathsf{RQSPACE_{\mathbb{R}}(1)} $}
			
			% CONNECT 3 to 3
			\multiput(59,31)(1,0){10}{$ . $} \put(68.3,31.2){\vector(1,0){1.5}}
			
			% CONNECT 3 to 4
			\put(75,33){\vector(0,1){11}}
			\multiput(60,32)(3,1){13}{$ . $} \put(98,45){\vector(3,1){1.5}}
			\multiput(110,37)(0,1){8}{$ . $} \put(110.2,43.5){\vector(0,1){1.5}}
			
			% LEVEL 5
			\put(7,60){$ \mathsf{BPSPACE_{\mathbb{R}}(1) } $}
			\put(45,60){$ \mathsf{PostBQAL} $}
			\put(65,60){$ \mathsf{ \frak{R}PostRQAL } $}
			\put(100,60){$ \mathsf{NQSPACE_{\mathbb{R}}(1)} $}
			\multiput(20,32)(0,1){31}{$ . $} \put(20.3,58){\vector(0,1){1.5}}
			
			% CONNECT 3 to 5
			\put(21,32){\vector(1,1){27.5}}
			\multiput(78,32)(1,1){27}{$ . $} \put(104,58){\vector(1,1){1.5}}
			\put(50,33){\vector(0,1){27}}
			
			% CONNECT 4 to 5
			\put(75,51){\vector(0,1){8.5}}
			\multiput(45,63)(-2,1){11}{$ . $} \put(25,73.5){\vector(-2,1){1.5}}
			\multiput(110,48)(0,1){11}{$ . $} \put(110.2,58.5){\vector(0,1){1.5}}
			
			% LEVEL 6
			\put(7,75){$ \mathsf{BQSPACE_{\mathbb{R}}(1) } $}
			\put(45,75){$ \mathsf{ QAL = S } $}
			\put(65,75){$ \mathsf{ \frak{R}PostBQAL } $}
			
			% CONNECT 5 to 6
			\put(20,63){\vector(0,1){11.5}}
			\put(50,63){\vector(0,1){11.5}}
			\multiput(75,63)(0,1){11}{$ . $} \put(75.2,73.5){\vector(0,1){1.5}}
			\multiput(55,63)(1,1){11}{$ . $} \put(65.4,73.5){\vector(1,1){1.5}}
			\put(25,63){\vector(2,1){21}}
			
			% LEVEL 7
			\put(65,90){$ \mathsf{ PrSPACE_{\mathbb{R}} (1) } $}
			\put(65,94){$ \mathsf{ uQAL = uS = } $}
			
			% CONNECT 6 to 7
			\multiput(75,78)(0,1){11}{$ . $} \put(75.2,88.5){\vector(0,1){1.5}}
			\multiput(56,78)(1,1){11}{$ . $} \put(66,88){\vector(1,1){1.5}}
			
			% LEVEL 8
			\put(65,105){$ \mathsf{ PrQSPACE_{\mathbb{R}} (1) } $}
			\put(75,96){\vector(0,1){9}}

			% CONNECT 5 to 8
			\multiput(103,63)(-1,2){22}{$ . $} \put(84,101){\vector(-1,2){1.5}}
			
			% CONNECT 6 to 8
			\multiput(28,78)(3,2){13}{$ . $} \put(66,103){\vector(3,2){1.5}}
			
			% CONNECT 7 to 8
			
		\end{picture}
		\end{center}
	\end{minipage}}
	\end{center}
	\caption{The relationships among standard and postselected versions of classical and quantum constant-memory classes}
	\vskip\baselineskip
	\label{fig:rt-rel}
\end{figure}

Defining the complexity classes using a more restricted model of QTMs, Watrous \cite{Wa99,Wa03} proved that $\mathsf{PrSPACE}(s)=  \mathsf{PrQSPACE}(s)$, and $\mathsf{NQSPACE}(s) = \mathsf{RQSPACE}(s)$ for all space-constructible $s=\Omega(\log n)$. We do not know if these equalities remains valid for the TM model used in this paper. Nor do we know whether quantum nondeterministic space is closed under complement or not, even in Watrous' restricted setup. (Note that a key step in our characterization of classical error-free space-bounded  classes was the closure of classical nondeterministic space under complementation.)

A fact that we know only for the $\mathsf{EPSPACE}$ family among the error-free classes is the existence of a hierarchy among them. This follows trivially from their identity with $\mathsf{NSPACE}$ classes, and the nondeterministic space hierarchy theorem. All other error-free  classes that we considered in Theorem \ref{thm:exactchar} were shown to equal the intersection of the corresponding ``nondeterministic" class and its complement. The classical nondeterministic time hierarchy is well known. Although we have not seen it stated anywhere, tight hierarchies for the nondeterministic quantum time- and space-bounded classes are also easily shown to exist, as we briefly describe below.

Using his QTM model mentioned above, Watrous \cite{Wa99} also proved for all space-constructible $s=\Omega(\log n)$ that $\mathsf{NQSPACE}(s) = \mathsf{coC_{=}SPACE}(s)$. \v{Z}\'{a}k's proof of the nondeterministic time hierarchy \cite{Za83,FS06} is general enough to apply to all of the $\mathsf{C_{=}SPACE}$, $\mathsf{coC_{=}SPACE}$, $\mathsf{C_{=}TIME}$, and $\mathsf{coC_{=}TIME}$ classes, where the last two families of classes are straightforward generalizations of $\mathsf{C_{=}P}$ and $\mathsf{coC_{=}P}$.\footnote{We thank Lance Fortnow for explaining this to us.} Finally, it is easy to generalize the proof that $\mathsf{NQP}=\mathsf{coC_{=}P}$ \cite{YY99} to hold for the families $\mathsf{NQTIME}$ and $\mathsf{coC_{=}TIME}$ as well.

Hierarchy theorems for intersection classes like the ones we have in the characterizations of  $\mathsf{EPTIME}$, $\mathsf{EQTIME}$, and $\mathsf{EQSPACE}$ given by Theorem \ref{thm:exactchar} are not presently known. Another important open question concerns the relationships among the various $\mathsf{TIME}$ and $\mathsf{QTIME}$ classes, and their postselected versions.
	
% SSSSSSSSSSSSSSSSSSSSSSSSSSSSSSSSSSSSSSSSSSSSSSSSSSSSSSSSSSSSSSSSSSSSSSSSSSSSSSSS %
% SSSSSSSSSSSSSSSSSSSSSSSSSSSSSSSSSSSSSSSSSSSSSSSSSSSSSSSSSSSSSSSSSSSSSSSSSSSSSSSS %
% SSSSSSSSSSSSSSSSSSSSSSSSSSSSSSSSSSSSSSSSSSSSSSSSSSSSSSSSSSSSSSSSSSSSSSSSSSSSSSSS %
\section*{Acknowledgements} \label{sec:Acknowledgement}
% SSSSSSSSSSSSSSSSSSSSSSSSSSSSSSSSSSSSSSSSSSSSSSSSSSSSSSSSSSSSSSSSSSSSSSSSSSSSSSSS %
% SSSSSSSSSSSSSSSSSSSSSSSSSSSSSSSSSSSSSSSSSSSSSSSSSSSSSSSSSSSSSSSSSSSSSSSSSSSSSSSS %
% SSSSSSSSSSSSSSSSSSSSSSSSSSSSSSSSSSSSSSSSSSSSSSSSSSSSSSSSSSSSSSSSSSSSSSSSSSSSSSSS %

We thank R\={u}si\c{n}\v{s} Freivalds for pointing us to the subject of this paper, 
and kindly providing us copies of references \cite{LSF09,SLF10}. We are grateful to Lance Fortnow for his  help about the hierarchies of counting classes, and to John Watrous, Scott Aaronson, Tomoyuki Yamakami, and Greg Kuperberg for their helpful answers to our questions. We also thank an anonymous referee for insightful comments on a previous version of this manuscript.

\appendix

% SSSSSSSSSSSSSSSSSSSSSSSSSSSSSSSSSSSSSSSSSSSSSSSSSSSSSSSSSSSSSSSSSSSSSSSSSSSSSSSS %
% SSSSSSSSSSSSSSSSSSSSSSSSSSSSSSSSSSSSSSSSSSSSSSSSSSSSSSSSSSSSSSSSSSSSSSSSSSSSSSSS %
% SSSSSSSSSSSSSSSSSSSSSSSSSSSSSSSSSSSSSSSSSSSSSSSSSSSSSSSSSSSSSSSSSSSSSSSSSSSSSSSS %
\section{Well-formedness of quantum machines} \label{app:wellform}
% SSSSSSSSSSSSSSSSSSSSSSSSSSSSSSSSSSSSSSSSSSSSSSSSSSSSSSSSSSSSSSSSSSSSSSSSSSSSSSSS %
% SSSSSSSSSSSSSSSSSSSSSSSSSSSSSSSSSSSSSSSSSSSSSSSSSSSSSSSSSSSSSSSSSSSSSSSSSSSSSSSS %
% SSSSSSSSSSSSSSSSSSSSSSSSSSSSSSSSSSSSSSSSSSSSSSSSSSSSSSSSSSSSSSSSSSSSSSSSSSSSSSSS %

For any input string $x$, let  $\mathcal{C}_x$ be a suitably ordered list of all the reachable configurations of the space-bounded QTM under consideration. The transition function $ \delta $ (see the definition in Section \ref{sec:Prel}) induces a set of $ | \mathcal{C}_x | \times | \mathcal{C}_x |$ configuration transition
matrices, $ \{ E_{\omega} \mid \omega \in \Omega \} $, where the $ (i,j)^{th} $ entry of $ E_{\omega}$,
the amplitude of the transition  from $ c_{j} $ to $ c_{i} $ by
writing $ \omega \in \Omega $ on the register,
is defined by $ \delta $ whenever the $j$th configuration $ c_{j} $ is reachable from the $i$th configuration $ c_{i}
$  in one step, and is zero otherwise. The QTM is said to be well-formed if 
\begin{equation}\label{equation:wf}
      \sum_{\omega \in \Omega} E_{\omega}^{\dagger}E_{\omega} = I.
\end{equation}

As described in \cite{YS11A}, an easy way of checking whether a QTM is well-formed is to verify if the columns of the $ | \mathcal{C}_x | |
\Omega | \times | \mathcal{C}_x | $-dimensional
matrix $ \mathsf{E} $ (Figure \ref{figure:matrix-E}) obtained by concatenating all the $ E_{\omega} $s one under the other,
form an orthonormal set. In the figure, and in the rest of this section, we use the convention that $\Delta=\{c,\tau_2,\tau_3,\ldots,\tau_k\}$ for $k>2$, $c$ is sometimes referred to as $\tau_1$, and the elements of $\Omega$ are ordered so that those corresponding to the same $\tau$ value are grouped together for each $\tau \in \Delta$.

\begin{center}
\begin{figure}[h!]
	\centering	
	\begin{minipage}{0.7\textwidth}
		\[
			\begin{array}{rl}
				& \begin{array}{rccccc} & ~~~~~ c_{1} & ~~ c_{2} &  ~ \ldots ~ & ~ c_{|\mathcal{C}_x|} \end{array}
				\\
				\tau_1 = c \left \lbrace \begin{array}{c} \\ \\ \\ \\ \\ \\ \\ \\ \\ \\ \\ \end{array} \right.
					& \mspace{-10mu}	
						\begin{array}{rccccc}
							\hline
							\multicolumn{1}{c|}{c_{1}} & & & & \multicolumn{1}{c|}{} \\
							\multicolumn{1}{c|}{c_{2}} & & & & \multicolumn{1}{c|}{} \\
							\multicolumn{1}{c|}{\vdots} & 
								\multicolumn{4}{c|}{ \hspace*{30pt} E_{\omega_{1}} \hspace*{29pt} }  \\
							\multicolumn{1}{c|}{c_{|\mathcal{C}_x|}} & & & & \multicolumn{1}{c|}{} \\
							\hline	
							\multicolumn{1}{c|}{c_{1}} & & & & \multicolumn{1}{c|}{} \\
							\multicolumn{1}{c|}{c_{2}} & & & & \multicolumn{1}{c|}{} \\
							\multicolumn{1}{c|}{\vdots} & \multicolumn{4}{c|}{ E_{\omega_{2}} } \\
							\multicolumn{1}{c|}{c_{|\mathcal{C}_x|}} & & & & \multicolumn{1}{c|}{} \\
							\hline
							\multicolumn{1}{c|}{\vdots} & \multicolumn{4}{c|}{\vdots } \\							
							\hline	
						\end{array}
				\\
				\tau_2 \left \lbrace \begin{array}{c} \\ \\ \\ \\ \\ \\  \end{array} \right.
					& \mspace{-10mu}
						\begin{array}{rccccc}
							\multicolumn{1}{c|}{c_{1}} & & & & \multicolumn{1}{c|}{} \\
							\multicolumn{1}{c|}{c_{2}} & & & & \multicolumn{1}{c|}{} \\
							\multicolumn{1}{c|}{\vdots} & 
								\multicolumn{4}{c|}{ \hspace*{30pt} E_{\omega_{|\Omega_c|+1}} \hspace*{10.5pt} } \\
							\multicolumn{1}{c|}{c_{|\mathcal{C}_x|}} & & & & \multicolumn{1}{c|}{} \\
							\hline
							\multicolumn{1}{c|}{\vdots} & \multicolumn{4}{c|}{\vdots } \\
							\hline	
						\end{array}
				\\
					& \mspace{-17mu}
						\begin{array}{rccccc}
						\multicolumn{1}{c|}{~~~\vdots~~~} & 
							\multicolumn{4}{c|}{ \hspace*{37pt} \vdots \hspace*{36pt} } \\
						\multicolumn{1}{c|}{~~~\vdots~~~} & 
							\multicolumn{4}{c|}{  \vdots  }
					\end{array}
				\\
				\tau_k \left \lbrace \begin{array}{c} \\ \\ \\ \\ \\ \\  \end{array} \right.
					& \mspace{-10mu}
						\begin{array}{rccccc}
							\hline
							\multicolumn{1}{c|}{\vdots} & \multicolumn{4}{c|}{\vdots } \\
							\hline	
							\multicolumn{1}{c|}{c_{1}} & & & & \multicolumn{1}{c|}{} \\
							\multicolumn{1}{c|}{c_{2}} & & & & \multicolumn{1}{c|}{} \\
							\multicolumn{1}{c|}{\vdots} & 
								\multicolumn{4}{c|}{ \hspace*{30pt} E_{\omega_{|\Omega|}} \hspace*{23pt} } \\
							\multicolumn{1}{c|}{c_{|\mathcal{C}_x|}} & & & & \multicolumn{1}{c|}{} \\
							\hline	
						\end{array}
			\end{array}
		\]
	\end{minipage}
	\caption{Matrix $ \mathsf{E} $}
	\label{figure:matrix-E}
\end{figure}
\end{center}

The QTM definition in \cite{YS11A} only requires that the finite register alphabet $\Omega$ is partitioned in terms of the possible outcomes in $\Delta$, and imposes no such condition on the state set $Q$, as we have been doing in this paper. We will now show that our additional restrictions do not cause a decrease of power, by describing a procedure for building a new QTM $ \mathcal{M}'$ which obeys these restrictions and recognizes the same language as a given QTM $ \mathcal{M}=(Q,\Sigma,\Gamma,\Omega,\delta,q_{1},\Delta) $, defined in  the more relaxed format of \cite{YS11A}:

Let   $ \mathcal{M}' = \{ Q',\Sigma,\Gamma,\Omega,\delta',(q_{1},\omega_1),\Delta\} $, where
\begin{itemize}
\item $ Q' = Q \times \Delta $, 
\item $ Q'_{\tau} = \{ (q,\tau) \mid q \in Q \} $ for all $\tau \in \Delta$, and,
\item $\omega_1$ is the initial register symbol of both $ \mathcal{M}$  and $ \mathcal{M}'$. 
\end{itemize}

We wish to define $ \delta' $ so that
it will be guaranteed that whenever a value $ \tau $ is observed as the result of a measurement of the register, all configurations with nonzero amplitudes will have as their state component  the value $ (q,\tau) $ for some $q \in Q$.
Thus, we can have a partition of the state set as well according to $ \Delta $.

In all the Turing machine variants in this paper, the only member of $\Delta$ that does \textit{not} indicate that the computation has halted is $c$. Transitions from configurations that the machine is in when any move outcome other than $c$ is observed will never actually be performed, since the machine will not be running any longer. Of course, the amplitudes corresponding to those transitions must still be selected so that Equation \ref{equation:wf} is satisfied. We make sure that the transitions from configurations with state component $ (q,c) $ for some $q \in Q$ mimic the corresponding transitions in $ \mathcal{M}$ while entering correctly categorized states of $ \mathcal{M}'$, and fill in the rest of $\delta'$ so that  $ \mathcal{M}'$ is well-formed, as described below.

For every $\tau \in \Delta$, if $ \mathcal{M} $ contains the transition 
\[\delta(q,\sigma,\gamma,q',d_i,\gamma',d_w,\omega) = \alpha, \]
where $ \omega \in \Omega_{\tau}$, we add the transition 
\[\delta'((q,c),\sigma,\gamma,(q',\tau),d_i,\gamma',d_w,\omega) = \alpha\]
to $ \mathcal{M}'$.

Consider the matrix of Figure \ref{figure:matrix-E} as representing the transitions of $ \mathcal{M}$ on an input $x$. For the new machine $ \mathcal{M}'$, we will have a bigger matrix, since the new list of reachable configurations is $ | \Delta |$ times longer than that of $ \mathcal{M}$. Ordering the configurations so that those with the same $\tau$ value in their state components are grouped together under the name $\mathcal{C}_x(Q_{\tau})$, we obtain the 
template of Figure \ref{figure:matrix-E-prime}, in which only the transitions from states $ (q,c) $ for $q \in Q$ are filled in, according to the specification presented above. 

\begin{center}
\begin{figure}[h!]
	\centering	
	\small
	\begin{minipage}{\textwidth}
		\[
			\begin{array}{rl}
				& \hspace*{40pt}
					\begin{array}{cccc}
						\mathcal{C}_x(Q_{c}) & ~~~ \mathcal{C}_x(Q_{\tau_2}) & \mspace{5mu}  \cdots & ~ \mathcal{C}_x(Q_{\tau_k})
					\end{array}
				\\
				\tau_1 = c \left \lbrace \begin{array}{c} \\ \\ \\ \\ \\ \\ \\ \\ \\ \\ \\ \\ \\ \end{array} \right.
					& \mspace{-10mu}
					\begin{array}{c|c|c|c|c|}
						\hline
						\hspace*{7pt}  \mathcal{C}_x(Q_{c}) \hspace*{6.5pt}  & 
							\hspace*{10pt}  E_{\omega_1} \hspace*{9.5pt}  & 
							\hspace*{33pt} & \hspace*{15pt} & \hspace*{33pt}
						\\ \hline
						\mathcal{C}_x(Q_{\tau_2}) & 0 &  & & 
						\\ \hline
						\vdots & \vdots & & & 
						\\ \hline 
						\mathcal{C}_x(Q_{\tau_k}) & 0  &  & &
						\\ \hline \hline
						\mathcal{C}_x(Q_{c}) & E_{\omega_2} &  & & 
						\\ \hline
						\mathcal{C}_x(Q_{\tau_2}) & 0 &  & & 
						\\ \hline
						\vdots & \vdots & & & 
						\\ \hline 
						\mathcal{C}_x(Q_{\tau_k}) & 0  &  & &
						\\ \hline \hline
						\vdots & \vdots & & & 
						\\
						\vdots & \vdots & & & 
						\\ \hline
					\end{array}
				\\ 
				\tau_2 \left \lbrace \begin{array}{c} \\ \\ \\ \\ \\ \\ \\ \\ \\ \end{array} \right.
					& \mspace{-10mu}
					\begin{array}{c|c|c|c|c|}
						\hline
						 \hspace*{7pt} \mathcal{C}_x(Q_{c}) \hspace*{6.5pt} & 0 & 
						 	\hspace*{33pt} & \hspace*{15pt} & \hspace*{33pt}
						\\ \hline
						\mathcal{C}_x(Q_{\tau_2}) &  E_{\omega_{|\Omega_c|+1}} \hspace*{1.5pt}  &  & & 
						\\ \hline
						\mathcal{C}_x(Q_{\tau_3}) &  0  &  & & 
						\\ \hline
						\vdots & \vdots & & & 
						\\ \hline 
						\mathcal{C}_x(Q_{\tau_k}) & 0  &  & &
						\\ \hline \hline
						\vdots & \vdots & & & 
						\\ \vdots & \vdots & & & 
						\\ \hline 
					\end{array}
				\\ 				
					& \mspace{-14mu}
					\begin{array}{c|c|c|c|c|}
						\hline
						\multicolumn{1}{c|}{ \hspace*{21pt} \vdots \hspace*{20pt}} & 
							\multicolumn{1}{c|}{ \hspace*{17pt} \vdots \hspace*{17.5pt} } 
							& \hspace*{33pt} & \hspace*{14.5pt} & \hspace*{33pt}
							 \\	
						\multicolumn{1}{c|}{\vdots} & \multicolumn{1}{c|}{\vdots} & & & \\	
						\multicolumn{1}{c|}{\vdots} & \multicolumn{1}{c|}{\vdots} & & & \\
						\hline
					\end{array}	
				\\ 
				\tau_k \left \lbrace \begin{array}{c} \\ \\ \\ \\ \\ \\ \\ \\ \end{array} \right.
					& \mspace{-10mu}
					\begin{array}{c|c|c|c|c|}
						\hline
						\vdots & \vdots & \hspace*{33pt} & \hspace*{15pt} & \hspace*{33pt}	
						\\
						\vdots & \vdots & & & 
						\\ \hline
						\hline
						\mathcal{C}_x(Q_{c}) & 0 &  & & 
						\\ \hline
						\vdots & \vdots & & & 
						\\ \hline 
						\mathcal{C}_x(Q_{\tau_{k-1}}) &  0  &  & & 
						\\ \hline			
						\mathcal{C}_x(Q_{\tau_k}) & \hspace*{7pt} E_{\omega_{|\Omega|}} \hspace*{7pt}  &  & &
						\\ \hline
					\end{array}
			\end{array}
		\]
	\end{minipage}
	\caption{Partially filled configuration transition matrix of $ \mathcal{M}' $ on input $x$}
	\label{figure:matrix-E-prime}
\end{figure}
\end{center}

The remainder of $\delta'$ is set as follows to ensure well-formedness:

For every $\tau_j,\tau_l \in \Delta$, if $ \mathcal{M} $ contains the transition 
\[\delta(q,\sigma,\gamma,q',d_i,\gamma',d_w,\omega) = \alpha, \]
where $ \omega \in \Omega_{\tau_l}$, we add the transition 
\[\delta'((q,j),\sigma,\gamma,(q',\tau_{j+l-1 ~(\mbox{mod } |\Delta|)}),d_i,\gamma',d_w,\omega) = \alpha\]
to $ \mathcal{M}'$.

This procedure yields a transition function that satisfies Equation \ref{equation:wf}, since it distributes exact copies of the $|\Omega|$ nonzero submatrices seen in the first ``column"\footnote{Note that each empty box in the matrix of Figure \ref{figure:matrix-E-prime} corresponds to a $ | \mathcal{C}_x | \times | \mathcal{C}_x |$ matrix.} of the matrix of Figure \ref{figure:matrix-E-prime} to $|\Omega|$ ``boxes" in each of the other columns, while ensuring that each ``row" contains exactly one such nonzero submatrix. The orthonormality of the actual set of column vectors follows from the well-formedness of $ \mathcal{M} $.

% SSSSSSSSSSSSSSSSSSSSSSSSSSSSSSSSSSSSSSSSSSSSSSSSSSSSSSSSSSSSSSSSSSSSSSSSSSSSSSSS %
% SSSSSSSSSSSSSSSSSSSSSSSSSSSSSSSSSSSSSSSSSSSSSSSSSSSSSSSSSSSSSSSSSSSSSSSSSSSSSSSS %
% SSSSSSSSSSSSSSSSSSSSSSSSSSSSSSSSSSSSSSSSSSSSSSSSSSSSSSSSSSSSSSSSSSSSSSSSSSSSSSSS %
\section{Definitions of postselection finite automata} \label{app:postfa}
% SSSSSSSSSSSSSSSSSSSSSSSSSSSSSSSSSSSSSSSSSSSSSSSSSSSSSSSSSSSSSSSSSSSSSSSSSSSSSSSS %
% SSSSSSSSSSSSSSSSSSSSSSSSSSSSSSSSSSSSSSSSSSSSSSSSSSSSSSSSSSSSSSSSSSSSSSSSSSSSSSSS %
% SSSSSSSSSSSSSSSSSSSSSSSSSSSSSSSSSSSSSSSSSSSSSSSSSSSSSSSSSSSSSSSSSSSSSSSSSSSSSSSS %

A real-time PFA with postselection (PFAP) is a 5-tuple 
$ \mathcal{P}=(Q,\Sigma,\{ A_{\sigma} \mid \sigma \in \tilde{\Sigma} \},q_{1},\Delta), $
where $ Q $,    
$ q_{1} $, 
$ \Sigma $ have the same semantics as in our previous definitions,
the $ A_{\sigma} $'s are  transition matrices, whose columns are
stochastic vectors, such that $ A_{\sigma} $'s $ (j,i)^{th}
$ entry, denoted  $ A_{\sigma}[j,i] $, is the probability of the transition from state $
q_{i} $ to state $ q_{j} $ when reading symbol $ \sigma $,
and $ \Delta = \{pa,pr,nh\} $.

The computation of a PFAP can be traced by a stochastic state vector,
say $ v $, whose $ i^{th} $ entry, denoted
$ v[i] $, corresponds to state $ q_{i} $.
For a given input string $ w \in \Sigma^{*} $ (the string read by the machine is $ \tilde{w} = \cent w \dollar $),
\begin{equation*}
%\label{equation:vAv}
      v_{i} = A_{\tilde{w}_{i}} v_{i-1},
\end{equation*}
where $ \tilde{w}_{i} $ denotes the $ i^{th} $ symbol of $ \tilde{w} $,  $ 1 \le i \le | \tilde{w} | $, and
$ v_{0} $ is the initial state vector, whose first entry is 1. ($ |\tilde{w}| $  denotes the length of $ \tilde{w} $.)
\begin{comment}
	The transition matrices of a PFAP can be extended for any string as
	\begin{equation*}
    	A_{w\sigma} = A_{\sigma} A_{w},
	\end{equation*}
	where $ w \in (\tilde{\Sigma})^{*} $, $ \sigma \in \tilde{\Sigma} $, and
	$ A_{\varepsilon} = I $ ($ \varepsilon $ denotes the empty string).
\end{comment}
A PFAP must satisfy the  inequality
\begin{equation*}
       \sum_{q_{i} \in \{ Q_{pa} \cup Q_{pr} \} } v_{|\tilde{w}|}[i] > 0
\end{equation*}
for all inputs $w$.

The \textit{acceptance and rejection probabilities of input string $w$ by PFAP $ \mathcal{P} $  
before postselection} are defined as
\begin{equation*}
       p_{\mathcal{P}}^{a}(w) = \sum_{q_{i} \in Q_{pa}} v_{|\tilde{w}|}[i]
       ~~~~ \mbox{ and } ~~~~
       p_{\mathcal{P}}^{r}(w) = \sum_{q_{i} \in Q_{pr}} v_{|\tilde{w}|}[i].
\end{equation*}

A realtime QFA with postselection (QFAP) is a 5-tuple
$ \mathcal{M}=(Q,\Sigma,\{\mathcal{E}_{\sigma} \mid \sigma \in \tilde{\Sigma}\},q_{1},\Delta), $
where $ Q$, $\Sigma$, $q_{1}$, and $ \Delta $ are as defined above for PFAPs, 
and $ \mathcal{E}_{\sigma } $ is a collection of state transition matrices\footnote{Note the difference from Appendix \ref{app:wellform}, where these were \textit{configuration} transition matrices.}
$ \{ E_{\sigma,1},\ldots,E_{\sigma,k} \} $ for some $ k \in \mathbb{Z}^{+} $
satisfying
\begin{equation*}
      \sum_{i=1}^{k} E_{\sigma,i}^{\dagger} E_{\sigma,i} = I.
\end{equation*}
Additionally, we define the  operator
\begin{equation*}
     P = \{ P_{ \tau \in \Delta } \mid P_{\tau} = \sum_{q \in Q_{\tau}} \ket{q}\bra{q} \}
\end{equation*}
representing the single measurement of the state type at the end of the computation.\footnote{Multiple measurements of the register are not required in real-time computation \cite{YS11A}. That is why we do not need any ``continuing" states in the partition induced by $\Delta$; the computation stops only at the end of the input.}

For a given input string $ w \in \Sigma^{*} $ (the string read by the machine is $ \tilde{w} = \cent w \dollar $), the overall state of the machine can be traced by
\begin{equation*}
      \rho_{j} = \mathcal{E}_{\tilde{w}_{j}} (\rho_{j-1}) =
      \sum_{i=1}^{k} E_{\tilde{w}_{j},i} \rho_{j-1}
E_{\tilde{w}_{j},i}^{\dagger},
\end{equation*}
where $ 1 \le j \le | \tilde{w} |  $ and $ \rho_{0} = \ket{q_{1}}
\bra{q_{1}} $ is the initial density matrix \cite{NC00}.
\begin{comment}
	The transition operators can  be extended easily for any string as
	\begin{equation*}
	      \mathcal{E}_{w \sigma} = \mathcal{E}_{\sigma} \circ \mathcal{E}_{w},
	\end{equation*}
	where $ w \in (\tilde{\Sigma})^{*} $, $ \sigma \in \tilde{\Sigma} $,
	and $ \mathcal{E}_{\varepsilon} = I $.
	(Note that
	$ \mathcal{E}^{\prime} \circ \mathcal{E} $ is described by the collection  $ \{ E^{\prime}_{j} E_{i}
	\mid 1 \le i \le k, 1 \le j \le k^{\prime} \} $,
	when $ \mathcal{E} $ and $ \mathcal{E}^{\prime} $ are described by the collections 
	$ \{E_{i} \mid 1 \le i \le k\} $ and $ \{E_{j}^{\prime} \mid 1 \le j \le k^{\prime}\} $, respectively.)
\end{comment}
A QFAP must satisfy the  inequality
\begin{equation*}
      tr(P_{pa} \rho_{| \tilde{w} |} ) + tr(P_{pr} \rho_{| \tilde{w} |} ) > 0
\end{equation*}
for all inputs $w$.

The \textit{acceptance and rejection probabilities of input string $w$ by QFAP 
$ \mathcal{M} $  before postselection} are defined as
\begin{equation*}
       p_{\mathcal{M}}^{a}(w) = tr(P_{pa} \rho_{| \tilde{w} |} )
        ~~~~ \mbox{ and } ~~~~
         p_{\mathcal{M}}^{r}(w) = tr(P_{pr} \rho_{| \tilde{w} |} ).
\end{equation*}

% SSSSSSSSSSSSSSSSSSSSSSSSSSSSSSSSSSSSSSSSSSSSSSSSSSSSSSSSSSSSSSSSSSSSSSSSSSSSSSSS %
% SSSSSSSSSSSSSSSSSSSSSSSSSSSSSSSSSSSSSSSSSSSSSSSSSSSSSSSSSSSSSSSSSSSSSSSSSSSSSSSS %
% SSSSSSSSSSSSSSSSSSSSSSSSSSSSSSSSSSSSSSSSSSSSSSSSSSSSSSSSSSSSSSSSSSSSSSSSSSSSSSSS %
\section{Error reduction for postselection finite automata} \label{app:postfa-error}
% SSSSSSSSSSSSSSSSSSSSSSSSSSSSSSSSSSSSSSSSSSSSSSSSSSSSSSSSSSSSSSSSSSSSSSSSSSSSSSSS %
% SSSSSSSSSSSSSSSSSSSSSSSSSSSSSSSSSSSSSSSSSSSSSSSSSSSSSSSSSSSSSSSSSSSSSSSSSSSSSSSS %
% SSSSSSSSSSSSSSSSSSSSSSSSSSSSSSSSSSSSSSSSSSSSSSSSSSSSSSSSSSSSSSSSSSSSSSSSSSSSSSSS %

Let $ \mathcal{M} $ be a machine with postselection (or, equivalently, a machine with restart).
We have the following relation \cite{YS10B}.
\begin{lemma}
	\label{lem:bounded-error}
    The language $ L \subseteq \Sigma^{*} $ is recognized by $ \mathcal{M} $ with error bound $ \epsilon $
	if and only if $ \frac{ p_{\mathcal{M}}^{r}(w) }{ p_{\mathcal{M}}^{a}(w) } 
	\le \frac{\epsilon}{1-\epsilon} $ when $ w \in L $, 
	and $ \frac{ p_{\mathcal{M}}^{a}(w) }{ p_{\mathcal{M}}^{r}(w) } \le \frac{\epsilon}{1-\epsilon} $
	when $ w \notin L $.
\end{lemma}
\begin{proof}
 This follows from Lemma \ref{lem:overall-acc-rej}, since, for all $ \epsilon \in [0,\frac{1}{2}) $,
 \begin{equation*}
 \label{equation:1overp-epsilon}
       \mbox{P(}\mathcal{M}\mbox{ accepts }w)	=\frac{ p_{\mathcal{M}}^{a} (w) }{ p_{\mathcal{M}}^{a} (w) + p_{\mathcal{M}}^{r} (w) } =\frac{1}{1+\frac{p_{\mathcal{M}}^{r} (w)}{p_{\mathcal{M}}^{a} (w)}} \ge 1-\epsilon \Leftrightarrow \frac{p_{\mathcal{M}}^{r} (w)}{p_{\mathcal{M}}^{a} (w)} \le
\frac{\epsilon}{1-\epsilon}.
 \end{equation*}
 The argument for $ w \notin L $ is identical.
\end{proof}

\begin{lemma}
	\label{lem:rt-amplification}
	If $ L $ is recognized by QFAP (resp., PFAP) $ \mathcal{M} $ with error bound 
	$ \epsilon \in (0,\frac{1}{2})  $,
	then there exists a QFAP (resp., PFAP), say $ \mathcal{M}' $,
	recognizing $ L $ with error bound $  \epsilon^{2} $.
\end{lemma}
\begin{proof}
	We give a proof for QFAPs, which can be adapted easily to PFAPs.
	$ M' $ can be obtained by taking the tensor product of $ k $ copies of $ \mathcal{M} $, where
	the new postselection accept (resp., reject)  states, $ Q_{pa}' $  (resp., $ Q_{pr}' $),
	are $ \otimes_{i=1}^{k} Q_{pa} $ (resp., $ \otimes_{i=1}^{k} Q_{pr} $),
	where $ Q_{pa} $ (resp., $ Q_{pr} $) are the postselection accept (resp., reject) states of $ \mathcal{M} $.
	
	Let $ \rho_{\tilde{w}} $ and $ \rho_{\tilde{w}}' $ be the  density matrices of 
	$ \mathcal{M} $ and $ \mathcal{M}' $, respectively,
	after reading $ \tilde{w} $ for a given input string $ w \in \Sigma^{*} $. 
	By definition, we have
	\begin{equation*}
		p_{\mathcal{M}}^{a}(w) = \sum_{q_{i} \in Q_{pa} }\rho_{\tilde{w}}[i,i],
		~~~~
		p_{\mathcal{M}'}^{a}(w) = \sum_{q_{i'} \in Q_{pa}' }\rho_{\tilde{w}}[i',i']		
	\end{equation*}
	and 
	 \begin{equation*}
		p_{\mathcal{M}}^{r}(w) = \sum_{q_{i} \in Q_{pr} }\rho_{\tilde{w}}[i,i],
		~~~~
		p_{\mathcal{M}'}^{r}(w) = \sum_{q_{i'} \in Q_{pr}' }\rho_{\tilde{w}}[i',i'].
	\end{equation*}
	By using the equality $ \rho_{\tilde{w}}' = \otimes_{i=1}^{k} \rho_{\tilde{w}}  $,
	the following can be obtained with a straightforward calculation:
	\begin{equation*}
		p_{\mathcal{M}'}^{a}(w) = \left( p_{\mathcal{M}}^{a}(w) \right)^{k}
	\end{equation*}
	and
	\begin{equation*}
		p_{\mathcal{M}'}^{r}(w) = \left( p_{\mathcal{M}}^{r}(w) \right)^{k}.
	\end{equation*}
	
	We examine the case of $ w \in L $ (the case $ w \notin L $ is symmetric).
	Since $ L $ is recognized by $ \mathcal{M} $ with error bound $ \epsilon $, 
	we have (due to Lemma \ref{lem:bounded-error})
	\begin{equation*}
		 \frac{p_{\mathcal{M}}^{r}(w)}{p_{\mathcal{M}}^{a}(w)}
		 \leq 
		 \frac{\epsilon}{1-\epsilon}.
	\end{equation*}
	If $ L $ is recognized by $ \mathcal{M}' $ with error bound $ \epsilon^{2} $,
	we must have
	\begin{equation*}
		\frac{p_{\mathcal{M}'}^{r}(w)}{p_{\mathcal{M}'}^{a}(w)}
		\leq
		\frac{\epsilon^{2}}{1-\epsilon^{2}}.
	\end{equation*}
	Thus, any $ k $ satisfying
	\begin{equation*}		
		\left( \frac{\epsilon}{1-\epsilon} \right)^{k}
		\leq		
		\frac{\epsilon^{2}}{1-\epsilon^{2}}
	\end{equation*}
provides the desired machine $ \mathcal{M}' $	due to the fact that 
	\begin{equation*}
		% \label{berr:eq:k}
		\frac{p_{\mathcal{M}'}^{r}(w)}{p_{\mathcal{M}'}^{a}(w)} =
		\left( \frac{p_{\mathcal{M}}^{r}(w)}{p_{\mathcal{M}}^{a}(w)} \right)^{k}.
	\end{equation*}
	By solving this equation, we  get the best value
	\begin{equation*}
		k = 1 + \left\lceil \frac{ \log \left( \frac{1}{\epsilon} + 1 \right) }{ 
		\log \left( \frac{1}{\epsilon} - 1 \right) } \right\rceil.
	\end{equation*}
	Therefore, for any $ 0 < \epsilon < \frac{1}{2} $, we can find a value for $ k $.
\end{proof}

\begin{theorem}
	\label{thm:rt-amplification}
	If $ L $ is recognized by QFAP (resp., PFAP) $ \mathcal{M} $ with error bound 
	$ 0 < \epsilon < \frac{1}{2}  $,
	then there exists a QFAP (resp., PFAP), say $ \mathcal{M}' $, recognizing $ L $ with error bound 
	$ \epsilon' < \epsilon $ such that $ \epsilon' $ can be arbitrarily close to 0.
\end{theorem}

% SSSSSSSSSSSSSSSSSSSSSSSSSSSSSSSSSSSSSSSSSSSSSSSSSSSSSSSSSSSSSSSSSSSSSSSSSSSSSSSS %
% SSSSSSSSSSSSSSSSSSSSSSSSSSSSSSSSSSSSSSSSSSSSSSSSSSSSSSSSSSSSSSSSSSSSSSSSSSSSSSSS %
% SSSSSSSSSSSSSSSSSSSSSSSSSSSSSSSSSSSSSSSSSSSSSSSSSSSSSSSSSSSSSSSSSSSSSSSSSSSSSSSS %
\section{Closure properties of $ \mathsf{PostBS} $  and $ \mathsf{PostBQAL} $} \label{app:postfa-closure}
% SSSSSSSSSSSSSSSSSSSSSSSSSSSSSSSSSSSSSSSSSSSSSSSSSSSSSSSSSSSSSSSSSSSSSSSSSSSSSSSS %
% SSSSSSSSSSSSSSSSSSSSSSSSSSSSSSSSSSSSSSSSSSSSSSSSSSSSSSSSSSSSSSSSSSSSSSSSSSSSSSSS %
% SSSSSSSSSSSSSSSSSSSSSSSSSSSSSSSSSSSSSSSSSSSSSSSSSSSSSSSSSSSSSSSSSSSSSSSSSSSSSSSS %

\begin{theorem}
	\label{thm:rt-post-closure}
	$ \mathsf{PostBS} $  and $ \mathsf{PostBQAL} $ are closed under complementation, union, and intersection.
\end{theorem}
\begin{proof}
	For any language recognized by a postselection finite automaton with bounded error, we can obtain a new machine
	recognizing the complement of that language with bounded error, 
	by just swapping the designations of the postselection accept and reject states.
	Therefore, both classes are closed under complementation.
	
	Let $ L_{1} $ and $ L_{2} $ be members of $ \mathsf{PostBQAL} $ (resp., $ \mathsf{PostBS} $).
	Then, there exist  two QFAPs (resp., PFAPs) $ \mathcal{P}_{1} $ and $ \mathcal{P}_{2} $	
	recognizing $ L_{1} $ and $ L_{2} $ with error bound $ \epsilon \leq \frac{1}{4} $, respectively.
	Moreover, let $ Q_{pa_{1}} $ and $ Q_{pr_{1}} $ (resp., $ Q_{pa_{2}} $ and $ Q_{pr_{2}} $)
	represent the sets of postselection accept and reject states of $ \mathcal{P}_{1} $
	(resp., $ \mathcal{P}_{2} $), respectively, and let $ Q_{p_{1}} =  Q_{pa_{1}} \cup Q_{pr_{1}} $
	and $ Q_{p_{2}} =  Q_{pa_{2}} \cup Q_{pr_{2}} $.
	By taking the tensor products of  $ \mathcal{P}_{1} $ and $ \mathcal{P}_{2} $, we obtain two new machines, 
	say $ \mathcal{M}_{1} $ and $ \mathcal{M}_{2} $, and set their definitions so that
	\begin{itemize}
		\item the sets of the postselection accept and reject  states of $ \mathcal{M}_{1} $ are
			$ Q_{p_{1}} \otimes Q_{p_{2}} \setminus Q_{pr_{1}} \otimes Q_{pr_{2}} $
			and
			$ Q_{pr_{1}} \otimes Q_{pr_{2}}, $ 
			respectively, and
		\item the sets of the postselection accept and reject  states of $ \mathcal{M}_{2} $ are
			$ Q_{pa_{1}} \otimes Q_{pa_{2}} $ and 
			$ Q_{p_{1}} \otimes Q_{p_{2}} \setminus Q_{pa_{1}} \otimes Q_{pa_{2}}, $
			respectively.
	\end{itemize}
	Thus, the following inequalities can be verified for a given input string $ w \in \Sigma^{*} $:
	\begin{itemize}
		\item if $ w \in L_{1} \cup L_{2} $, $ \mbox{P(}\mathcal{M}_{1}\mbox{ accepts }w) \ge \frac{15}{16} $;
		\item if $ w \notin L_{1} \cup L_{2} $, $ \mbox{P(}\mathcal{M}_{1}\mbox{ accepts }w) \leq \frac{7}{16} $;
		\item if $ w \in L_{1} \cap L_{2} $, $ \mbox{P(}\mathcal{M}_{2}\mbox{ accepts }w) \ge \frac{9}{16} $; 
		\item if $ w \notin L_{1} \cap L_{2} $, $ \mbox{P(}\mathcal{M}_{2}\mbox{ accepts }w) \leq \frac{1}{16} $.
	\end{itemize}
	We conclude that both classes are closed under union and intersection.
\end{proof}

% SSSSSSSSSSSSSSSSSSSSSSSSSSSSSSSSSSSSSSSSSSSSSSSSSSSSSSSSSSSSSSSSSSSSSSSSSSSSSSSS %
% SSSSSSSSSSSSSSSSSSSSSSSSSSSSSSSSSSSSSSSSSSSSSSSSSSSSSSSSSSSSSSSSSSSSSSSSSSSSSSSS %
% SSSSSSSSSSSSSSSSSSSSSSSSSSSSSSSSSSSSSSSSSSSSSSSSSSSSSSSSSSSSSSSSSSSSSSSSSSSSSSSS %
\section{More results on Riga postselection finite automata} \label{app:lpostfa}
% SSSSSSSSSSSSSSSSSSSSSSSSSSSSSSSSSSSSSSSSSSSSSSSSSSSSSSSSSSSSSSSSSSSSSSSSSSSSSSSS %
% SSSSSSSSSSSSSSSSSSSSSSSSSSSSSSSSSSSSSSSSSSSSSSSSSSSSSSSSSSSSSSSSSSSSSSSSSSSSSSSS %
% SSSSSSSSSSSSSSSSSSSSSSSSSSSSSSSSSSSSSSSSSSSSSSSSSSSSSSSSSSSSSSSSSSSSSSSSSSSSSSSS %

\begin{theorem}
	\label{thm:lpostfa-equal}
	$ \mathsf{\frak{R}PostBS} $ = $ \mathsf{PostBS} $, $ \mathsf{\frak{R}PostRS} $ = $ \mathsf{PostRS} $, and $ \mathsf{PostES} $=$ \mathsf{\frak{R}PostES} $.
\end{theorem}
\begin{proof}
	We give a proof of the first equality. The same technique can also be used for the remaining equalities.
	
	Since  $ \mathsf{PostBS} $ $ \subseteq $ $ \mathsf{\frak{R}PostBS} $ is trivial, 
	we show that $ \mathsf{\frak{R}PostBS} $ $ \subseteq $ $ \mathsf{PostBS} $.
	Let $ L $ be in $ \mathsf{\frak{R}PostBS} $, and let $ \mathcal{P} $ with state set $Q$, postselection states 
	$Q_{p}= Q_{pa} \cup  Q_{pr} $, and $ \chi \in \{A,R\} $ be the Riga PFA with postselection
	recognizing $ L $ with error bound $ \epsilon < \frac{1}{2} $.
	Suppose that $ L' $ is the set of strings that lead 
	$ \mathcal{P} $ to the postselection set with zero probability.
	By designating all postselection states as accepting states and removing the probability
	values of transitions, we obtain a real-time nondeterministic finite automaton 
	which recognizes the complement of $L'$.
	Thus, there exists a real-time deterministic finite automaton, say $ \mathcal{D} $, recognizing $ L' $. 
	Let $Q_{ \mathcal{D} }$ and $A_{\mathcal{D} }$ be the overall state set, and the set of accept states of 
	$ \mathcal{D} $, respectively.

	We combine $ \mathcal{P} $ and $ \mathcal{D} $ with a tensor product to obtain a PFAP $ \mathcal{P}' $. 
	The postselection state set of $ \mathcal{P}' $ is 
	$((Q \setminus Q_{p})\otimes A_{\mathcal{D}}) \cup (Q_{p}\otimes(Q_{ \mathcal{D}} \setminus A_{\mathcal{D}}))$. 
	The postselection accept states of $ \mathcal{P}' $ are:
	\begin{equation*}
		\left\lbrace
		\begin{array}{lll}
			((Q \setminus Q_{p})\otimes A_{\mathcal{D}}) \cup (Q_{pa}\otimes(Q_{ \mathcal{D}} \setminus A_{\mathcal{D}}))
				& , ~~ &\mbox{if } \chi = ``A", \\
			Q_{pa}\otimes(Q_{ \mathcal{D}} \setminus A_{\mathcal{D}}) & , &\mbox{if } \chi = ``R".
		\end{array}
		\right.
	\end{equation*}
	$ \mathcal{P}' $ is structured so that 
	if the input string $w$ is in $ L' $, 
	the decision is given deterministically with respect to $ \chi $, and
	if  $w \notin L' $, (that is, the probability of postselection by $ \mathcal{P} $ is nonzero,)
	the decision is given by the standard postselection procedure.
	Therefore, $ L $ is recognized by $ \mathcal{P}' $ with the same error bound as  $ \mathcal{P} $, 
	meaning that $ L \in \mathsf{PostBS}$.
\end{proof}

\begin{theorem}
	$ \mathsf{\frak{R}PostEQAL} $ $ \subsetneq $ $ \mathsf{\frak{R}PostRQAL} $.
\end{theorem}
\begin{proof}	
	Recall the language $ L_{eq\overline{eq}}  = \{ aw_{1} \cup bw_{2} \mid w_{1} \in L_{eq}, w_{2} \in \overline{L_{eq}} \} $ from Theorem \ref{thm:2posqbeats1q}. $ L_{eq\overline{eq}}$ is not a member of $ \mathsf{NQAL} $ $ \cup $ $ \mathsf{coNQAL}  $ \cite{YS10A}, which equals $ \mathsf{\frak{R}PostEQAL} $ by Theorem \ref{thm:nuncon}.

We know that $ \overline{L_{eq}} $ is a member of  $ \mathsf{PostRQAL} $ (\cite{YS10B} and Theorem \ref{thm:posres}).
Let $ \mathcal{M}_1 $  be the corresponding machine. We first build a QFAP called $\mathcal{M}_2 $ by
converting all postselection accept states of $ \mathcal{M}_1 $ to postselection reject states, and setting all remaining states of $ \mathcal{M}_1 $ to be nonpostselection halting states.

Based on these machines, we construct a RQFAP $ \mathcal{R} $
recognizing $ L_{eq\overline{eq}} $ with one-sided bounded error, as follows:

If the input length is less than two, $ \mathcal{R} $ gives memorized answers.
Therefore, we assume that length of the input is greater than 1 in the remainder.
Let $ u $ be the postfix of the input obtained by deleting the first symbol.
If the first input symbol is a $ b $, 
$ \mathcal{R} $ passes control to $ \mathcal{M}_1 $.
Thus if $ u \in \overline{L_{eq}} $, the input is accepted with a probability greater than $ \frac{1}{2} $,
and if $ u \notin \overline{L_{eq}} $, the input is  rejected with certainty.

If the first input symbol is an $ a $, 
$ \mathcal{R} $ passes control to $ \mathcal{M}_2 $.
Thus, if $ u \notin L_{eq} $, then the input is  rejected with certainty, and
if $ u \in L_{eq} $, the computation ends in some nonpostselection halting states 
with probability 1. Therefore, the decision is given by the value of $ \chi $.
So, if we set the $ \chi $ of $ \mathcal{R} $ to ``$ A $", $ L_{eq\overline{eq}} $ is recognized with one-sided bounded error, as required.
\end{proof}

\begin{theorem}
	\label{thm:RQFAP-closure}
	$ \mathsf{\frak{R}PostBQAL} $ is closed under complementation.
\end{theorem}
\begin{proof}
	If a language is recognized by a RQFAP with bounded error,
	by swapping the accepting and rejecting postselection states and
	by setting $ \chi $ to $ \{A,R\} \setminus \chi $, we obtain a new RQFAP
	recognizing the complement of the language with bounded error.
	Therefore, $ \mathsf{\frak{R}PostBQAL} $ is closed under complementation.
\end{proof}

\begin{theorem}
	\label{thm:RPostBQAL-subset-uQAL}
	$ \mathsf{\frak{R}PostBQAL} $ $ \subseteq $ $ \mathsf{uQAL}=\mathsf{uS} $.
\end{theorem}
\begin{proof}
The equality has been shown in \cite{YS11A}.
	The rest of the proof is similar to the proof of Theorem \ref{thm:PostQ-subset-Q}, with the exception that
	\begin{itemize}
		\item if $ \chi = A $, we have recognition with nonstrict cutpoint;
		\item if $ \chi = R $, we have recognition with strict cutpoint.
	\end{itemize}
\end{proof}

It was shown in \cite{DF10} that the language $ L_{say} $, i.e.
$
	\{ w  \mid \exists u_{1},u_{2},v_{1},v_{2} \in \{a,b\}^{*},
			 w = u_{1}bu_{2} = v_{1}bv_{2}, |u_{1}| = |v_{2}| \},
$
cannot be recognized by a RQFAP. 
Since $ L_{say} \notin \mathsf{uS} $ \cite{FYS10A}, the same result follows easily from Theorem 
\ref{thm:RPostBQAL-subset-uQAL}.

\bibliographystyle{alpha}
\bibliography{YakaryilmazSay}

\begin{thebibliography}{BJKP05}

\bibitem[Aan74]{Aa74}
S.~O. Aanderra.
\newblock On k-tape versus (k-1)-tape real time computation.
\newblock In R.~M. Karp, editor, {\em SIAM AMS Proceedings}, volume 7
  (Complexity of Computation), pages 75--96, 1974.

\bibitem[Aar05]{Aa05}
Scott Aaronson.
\newblock Quantum computing, postselection, and probabilistic polynomial-time.
\newblock {\em Proceedings of the Royal Society A}, 461(2063):3473--3482, 2005.

\bibitem[ADH97]{ADH97}
Leonard~M. Adleman, Jonathan DeMarrais, and Ming-Deh~A. Huang.
\newblock Quantum computability.
\newblock {\em SIAM Journal on Computing}, 26(5):1524--1540, 1997.

\bibitem[AGM92]{AGM92}
Helmut Alt, Viliam Geffert, and Kurt Mehlhorn.
\newblock A lower bound for the nondeterministic space complexity of
  context-free recognition.
\newblock {\em Information Processing Letters}, 42(1):25--27, 1992.

\bibitem[AW02]{AW02}
Andris Ambainis and John Watrous.
\newblock Two--way finite automata with quantum and classical states.
\newblock {\em Theoretical Computer Science}, 287(1):299--311, 2002.

\bibitem[BJKP05]{BJKP05}
Vincent~D. Blondel, Emmanuel Jeandel, Pascal Koiran, and Natacha Portier.
\newblock Decidable and undecidable problems about quantum automata.
\newblock {\em SIAM Journal on Computing}, 34(6):1464--1473, 2005.

\bibitem[BM80]{BM80}
Anna~R. Bruss and Albert~R. Meyer.
\newblock On time-space classes and their relation to the theory of real
  addition.
\newblock {\em Theoretical Computer Science}, 11(1):59--69, May 1980.

\bibitem[Boz03]{Bo03}
Symeon Bozapalidis.
\newblock Extending stochastic and quantum functions.
\newblock {\em Theory of Computing Systems}, 36(2):183--197, 2003.

\bibitem[Bru02]{Br02}
Stefan~D. Bruda.
\newblock {\em Parallel Real-Time Complexity Theory}.
\newblock PhD thesis, Queen's University at Kingston, 2002.

\bibitem[BV97]{BV97}
Ethan Bernstein and Umesh Vazirani.
\newblock Quantum complexity theory.
\newblock {\em SIAM Journal on Computing}, 26(5):1411--1473, 1997.

\bibitem[BW11]{BW11}
Todd~A. Brun and Mark~M. Wilde.
\newblock Perfect state distinguishability and computational speedups with
  postselected closed timelike curves.
\newblock Technical report, arXiv:1008.0433, 2011.

\bibitem[DS90]{DS90}
Cynthia Dwork and Larry Stockmeyer.
\newblock A time complexity gap for two-way probabilistic finite-state
  automata.
\newblock {\em SIAM Journal on Computing}, 19(6):1011--1123, 1990.

\bibitem[DvM06]{DM06}
Scott Diehl and Dieter van Melkebeek.
\newblock Time-space lower bounds for the polynomial-time hierarchy on
  randomized machines.
\newblock {\em SIAM Journal on Computing}, 36(3):563--594, 2006.

\bibitem[FK94]{FK94}
R\={u}si\c{n}\v{s} Freivalds and Marek Karpinski.
\newblock Lower space bounds for randomized computation.
\newblock In {\em ICALP'94: Proceedings of the 21st International Colloquium on
  Automata, Languages and Programming}, pages 580--592, 1994.

\bibitem[Fli72]{Fl72}
Michel Fliess.
\newblock Automates stochastiques et s{\'e}ries rationnelles non commutatives.
\newblock In {\em Automata, Languages, and Programming}, pages 397--411, 1972.

\bibitem[Fli74]{Fl74}
Michel Fliess.
\newblock Propri\'{e}t\'{e}s bool\'{e}ennes des langages stochastiques.
\newblock {\em Mathematical Systems Theory}, 7(4):353--359, 1974.

\bibitem[Fre85]{Fr85}
R\={u}sin\v{s} Freivalds.
\newblock Space and reversal complexity of probabilistic one-way turing
  machines.
\newblock {\em Annals of Discrete MAthematics}, 24:39--50, 1985.

\bibitem[FS06]{FS06}
Lance Fortnow and Rahul Santhanam.
\newblock Recent work on hierarchies for semantic classes.
\newblock {\em ACM SIGACT News}, 37(3):36--54, 2006.

\bibitem[FYS10]{FYS10A}
R\={u}si\c{n}\v{s} Freivalds, Abuzer Yakary{\i}lmaz, and A.~C.~Cem Say.
\newblock A new family of nonstochastic languages.
\newblock {\em Information Processing Letters}, 110(10):410--413, 2010.

\bibitem[Gil77]{Gi77}
John Gill.
\newblock Computational complexity of probabilistic \mbox{Turing} machines.
\newblock {\em SIAM Journal on Computing}, 6(4):675--695, 1977.

\bibitem[Gol08]{Go08}
Oded Goldreich.
\newblock {\em Computational Complexity: A Conceptual Perspective}.
\newblock Cambridge University Press, 2008.

\bibitem[Jea07]{Je07}
Emmanuel Jeandel.
\newblock Topological automata.
\newblock {\em Theory of Computing Systems}, 40(4):397--407, 2007.

\bibitem[Ka{\c{n}}91]{Ka91}
J\={a}nis Ka{\c{n}}eps.
\newblock Stochasticity of the languages acceptable by two-way finite
  probabilistic automata.
\newblock {\em Discrete Mathematics and Applications}, 1:405--421, 1991.

\bibitem[KW97]{KW97}
Attila Kondacs and John Watrous.
\newblock On the power of quantum finite state automata.
\newblock In {\em FOCS'97: Proceedings of the 38th Annual Symposium on
  Foundations of Computer Science}, pages 66--75, 1997.

\bibitem[Lap74]{La74}
J\={a}nis Lapi\c{n}\v{s}.
\newblock On nonstochastic languages obtained as the union and intersection of
  stochastic languages.
\newblock {\em Avtom. Vychisl. Tekh.}, (4):6--13, 1974.
\newblock (Russian).

\bibitem[LSF09]{LSF09}
Lelde L\={a}ce, Oksana {Scegulnaja-Dubrovska}, and R\={u}si\c{n}\v{s}
  Freivalds.
\newblock Languages recognizable by quantum finite automata with cut-point 0.
\newblock In {\em SOFSEM'09: Proceedings of the 35th International Conference
  on Current Trends in Theory and Practice of Computer Science}, volume~2,
  pages 35--46, 2009.

\bibitem[NC00]{NC00}
Michael~A. Nielsen and Isaac~L. Chuang.
\newblock {\em Quantum Computation and Quantum Information}.
\newblock Cambridge University Press, 2000.

\bibitem[Rab63a]{Ra63}
Michael~O. Rabin.
\newblock Probabilistic automata.
\newblock {\em Information and Control}, 6:230--243, 1963.

\bibitem[Rab63b]{Ra63B}
Michael~O. Rabin.
\newblock Real time computation.
\newblock {\em Israel Journal of Mathematics}, 1(4):203--211, 1963.

\bibitem[Sak96]{Sa96}
Michael Saks.
\newblock Randomization and derandomization in space-bounded computation.
\newblock In {\em Proceedings of the 11th Annual IEEE Conference on
  Computational Complexity}, pages 128--149, 1996.

\bibitem[SF10]{DF10}
Oksana {Scegulnaja-Dubrovska} and R\={u}si\c{n}\v{s} Freivalds.
\newblock A context-free language not recognizable by postselection finite
  quantum automata.
\newblock In R\={u}si\c{n}\v{s} Freivalds, editor, {\em Randomized and quantum
  computation}, pages 35--48, 2010.
\newblock Satellite workshop of MFCS and CSL 2010.

\bibitem[Sho97]{Sh97}
Peter~W. Shor.
\newblock Polynomial-time algorithms for prime factorization and discrete
  logarithms on a quantum computer.
\newblock {\em SIAM Journal on Computing}, 26(5):1484--1509, 1997.

\bibitem[Sip06]{Si06}
Michael Sipser.
\newblock {\em Introduction to the Theory of Computation, 2nd edition}.
\newblock Thomson Course Technology, United States of America, 2006.

\bibitem[SLF10]{SLF10}
Oksana {Scegulnaja-Dubrovska}, Lelde L\={a}ce, and R\={u}si\c{n}\v{s}
  Freivalds.
\newblock Postselection finite quantum automata.
\newblock In {\em Unconventional Computation}, volume 6079 of LNCS of {\em
  LNCS}, pages 115--126, 2010.

\bibitem[Tur82]{Tu82}
Paavo Turakainen.
\newblock {\em Discrete Mathematics}, volume~7 of {\em Banach Center
  Publications}, chapter Rational stochastic automata in formal language
  theory, pages 31--44.
\newblock PWN-Polish Scientific Publishers, Warsaw, 1982.

\bibitem[vMW08]{MW08}
Dieter van Melkebeek and Thomas Watson.
\newblock A quantum time-space lower bound for the counting hierarchy.
\newblock {\em Electronic Colloquium on Computational Complexity (ECCC)},
  15(017), 2008.
\newblock Available at
  http://eccc.hpi-web.de/eccc-reports/2008/TR08-017/index.html.

\bibitem[\v{Z}83]{Za83}
Stanislav \v{Z}\'{a}k.
\newblock A \mbox{Turing} machine time hierarchy.
\newblock {\em Theoretical Computer Science}, 26(3):327--333, 1983.

\bibitem[Wat99]{Wa99}
John Watrous.
\newblock Space-bounded quantum complexity.
\newblock {\em Journal of Computer and System Sciences}, 59(2):281--326, 1999.

\bibitem[Wat03]{Wa03}
John Watrous.
\newblock On the complexity of simulating space-bounded quantum computations.
\newblock {\em Computational Complexity}, 12(1-2):48--84, 2003.

\bibitem[Wat09]{Wa09}
John Watrous.
\newblock Quantum computational complexity.
\newblock In Robert~A. Meyers, editor, {\em Encyclopedia of Complexity and
  Systems Science}, pages 7174--7201. Springer, 2009.

\bibitem[Yao93]{Ya93}
Andrew Chi-Chih Yao.
\newblock Quantum circuit complexity.
\newblock In {\em SFCS'93: Proceedings of the 1993 IEEE 34th Annual Foundations
  of Computer Science}, pages 352--361, 1993.

\bibitem[YS10a]{YS10A}
Abuzer Yakary{\i}lmaz and A.~C.~Cem Say.
\newblock Languages recognized by nondeterministic quantum finite automata.
\newblock {\em Quantum Information and Computation}, 10(9\&10):747--770, 2010.

\bibitem[YS10b]{YS10D}
Abuzer Yakary{\i}lmaz and A.~C.~Cem Say.
\newblock Probabilistic and quantum finite automata with postselection.
\newblock In R\={u}si\c{n}\v{s} Freivalds, editor, {\em Randomized and quantum
  computation}, pages 14--24, 2010.
\newblock Satellite workshop of MFCS and CSL 2010.

\bibitem[YS10c]{YS10B}
Abuzer Yakary{\i}lmaz and A.~C.~Cem Say.
\newblock Succinctness of two-way probabilistic and quantum finite automata.
\newblock {\em Discrete Mathematics and Theoretical Computer Science},
  12(2):19--40, 2010.

\bibitem[YS11a]{YS11B}
Abuzer Yakary{\i}lmaz and A.~C.~Cem Say.
\newblock Probabilistic and quantum finite automata with postselection.
\newblock Technical report, 2011.
\newblock arXiv:1102.0666.

\bibitem[YS11b]{YS11A}
Abuzer Yakary{\i}lmaz and A.~C.~Cem Say.
\newblock Unbounded-error quantum computation with small space bounds.
\newblock {\em Information and Computation}, 209(6):873--892, 2011.

\bibitem[YY99]{YY99}
Tomoyuki Yamakami and Andrew Chi-Chih Yao.
\newblock $ \mbox{NQP}_{ \mathbb{C} } $ = co-$ \mbox{C}_{=}\mbox{P}$.
\newblock {\em Information Processing Letters}, 71(2):63--69, 1999.

\end{thebibliography}

\end{document}